\documentclass[10pt, letterpaper, onecolumn, conference]{IEEEtran}

\usepackage{amsmath}
\usepackage[noadjust]{cite}
\usepackage{amsthm}
\usepackage{amssymb}
\usepackage[hypertexnames=false]{hyperref}

\usepackage{graphicx}
\usepackage{paralist}
\usepackage{booktabs}
\usepackage[noend]{algpseudocode}
\algnewcommand{\lIf}[1]{\State\algorithmicif\ #1\ \algorithmicthen}
\algblockdefx[EElse]{EElse}{EndEElse}[1][Unknown]{\algorithmicelse}{}
\makeatletter
\ifthenelse{\equal{\ALG@noend}{t}}%
{\algtext*{EndEElse}}
{}%
\makeatother
\newcommand{\setalglineno}[1]{%
	  \setcounter{ALG@line}{\numexpr#1-1}}

\newtheorem{lemma}{Lemma}
\newtheorem{theorem}{Theorem}

\newtheorem{remarkb}{Remark}
\newtheorem{proposition}{Proposition}

\newtheorem{claim}{Claim}
\newtheorem{observation}{Observation}

\newcommand{\para}[1]{\smallskip\noindent\emph{#1}}
\newcommand{\ls}{\langle}
\newcommand{\rs}{\rangle}
\newcommand{\D}{\mathcal{D}}
\newcommand{\Continue}{\textbf{continue}}

\newcommand{\SCCs}[1]{\ensuremath{\mathsf{SCC}(#1)}}
\newcommand{\EC}[1]{\ensuremath{\mathsf{EC}(#1)}}
\newcommand{\AEC}[1]{\ensuremath{\mathsf{AEC}(#1)}}
\newcommand{\Out}[1]{\ensuremath{\mathit{Out}(#1)}}
\newcommand{\In}[1]{\ensuremath{\mathit{In}(#1)}}
\newcommand{\Attr}[3]{\ensuremath{\mathsf{Attr}_{#1}^{#2}(#3)}}

\newcommand{\Pre}[2]{\ensuremath{\mathsf{Pre}_{#1}(#2)}}
\newcommand{\Post}[2]{\ensuremath{\mathsf{Post}_{#1}(#2)}}
\newcommand{\Set}[1]{\ensuremath{\{#1}\}}
\newcommand{\SCCFind}[2]{\ensuremath{\mathsf{SCC\mbox{-}Find}_{#1}(#2)}}
\newcommand{\Diam}[1]{\ensuremath{\mathsf{diam}(#1)}}
\newcommand{\Pick}[1]{\ensuremath{\mathsf{Pick}(#1)}}
\newcommand{\ROut}[1]{\ensuremath{\mathsf{ROut}(#1)}}

\newcommand{\Parity}[1]{\ensuremath{\mathrm{Parity(#1)}}}
\newcommand{\Reach}[1]{\ensuremath{\mathrm{Reachability(#1)}}}
\newcommand{\GraphReach}[2]{\ensuremath{\mathsf{GraphReach}_{#1}(#2)}}

\newcommand{\minPriority}[1]{\ensuremath{\mathsf{MinPriority}_(#1)}}

\newcommand{\Inf}[1]{\ensuremath{\mathrm{Inf(#1)}}}
\newcommand{\ASW}[1]{\ensuremath{\langle\!\langle1\rangle\!\rangle_{\mathit{almost}}(#1)}}
\newcommand{\WE}{\ensuremath{\mathsf{WE}}}

\newcommand{\A}{\ensuremath{\mathcal{A}}}

\renewcommand{\O}{\widetilde{O}}
\renewcommand{\P}{\ensuremath{\mathcal{P}}}
\newcommand{\aspace}{\ensuremath{\mathit{space}}}

\begin{document}

\title{Symbolic Time and Space Tradeoffs for Probabilistic Verification}
\author{\IEEEauthorblockN{Krishnendu Chatterjee}
	\IEEEauthorblockA{IST Austria, Austria}
	\and
	\IEEEauthorblockN{Wolfgang Dvo\v{r}\'ak}
	\IEEEauthorblockA{Institute of Logic and Computation\\ TU Wien, Austria}
	\and
	\IEEEauthorblockN{Monika Henzinger and Alexander Svozil}
	\IEEEauthorblockA{Theory and Application of Algorithms\\ University of Vienna, Austria}}

\IEEEoverridecommandlockouts
	
\maketitle

\begin{abstract}
We present a faster \emph{symbolic} algorithm for the following central
problem in
probabilistic verification: Compute the maximal end-component (MEC)
decomposition
of Markov decision processes (MDPs).
This problem generalizes the SCC decomposition problem of graphs and closed
recurrent sets of Markov chains.
The model of symbolic algorithms is widely used in formal verification and
model-checking, where access to the input model is restricted to only
symbolic operations (e.g., basic set operations and computation of one-step
neighborhood).
For an input MDP with $n$ vertices and $m$ edges, the classical symbolic
algorithm from the 1990s for  the MEC decomposition requires $O(n^2)$
symbolic operations and $O(1)$ symbolic space.
The only other symbolic algorithm for the MEC decomposition requires
$O(n \sqrt{m})$ symbolic operations and $O(\sqrt{m})$ symbolic space.
The main open question has been whether the worst-case $O(n^2)$ bound for
symbolic operations can be beaten for MEC decomposition computation.
In this work, we answer the open question in affirmative.
We present a symbolic algorithm that requires $\widetilde{O}(n^{1.5})$
symbolic
operations and $\widetilde{O}(\sqrt{n})$ symbolic space.
Moreover, the parametrization of our algorithm provides a trade-off between symbolic operations and symbolic space:
for all $0<\epsilon \leq 1/2$ the symbolic algorithm requires
$\widetilde{O}(n^{2-\epsilon})$ symbolic operations and
$\widetilde{O}(n^{\epsilon})$ symbolic space ($\widetilde{O}(\cdot)$ hides poly-logarithmic factors).

Using our techniques we also present faster algorithms for computing the almost-sure winning regions of $\omega$-regular objectives
for MDPs.
We consider the canonical parity objectives for $\omega$-regular objectives, and for parity objectives
with $d$-priorities we present an algorithm that computes the almost-sure winning region with 
$\widetilde{O}(n^{2-\epsilon})$ symbolic operations and
$\widetilde{O}(n^{\epsilon})$ symbolic space, for all $0 < \epsilon \leq 1/2$.
In contrast, previous approaches require either (a)~$O(n^2\cdot d)$ symbolic operations and $O(\log
n)$ symbolic space;
or (b)~$O(n \sqrt{m} \cdot d)$ symbolic operations and $O(\sqrt{m})$ symbolic space.
Thus we improve the time-space product from $\widetilde{O}(n^2\cdot d)$ to $\widetilde{O}(n^2)$.
\end{abstract}

\section{Introduction}
The verification of probabilistic systems, e.g., randomized protocols, 
or agents in uncertain environments like robot planning is a fundamental
problem in formal methods. 
We study a classical graph algorithmic problem that arises in the verification of
probabilistic systems and present a faster symbolic algorithm for it.
We start with the description of the graph problem and its applications, then 
describe the symbolic model of computation, then previous results, and finally our 
contributions.

\para{MEC decomposition.}
Given a finite directed graph $G=(V,E)$ with a set $V$ of vertices, a set 
$E \subseteq V \times V$ of directed edges, and a partition $(V_1,V_R)$ of $V$,
a nontrivial \emph{end-component (EC)}  is a set $U \subseteq V$ of vertices such that 
(a)~the graph $(U, E \cap U \times U)$ is strongly connected; 
(b)~for all $u \in U \cap V_R$ and all $(u,v) \in E$ we have $v\in  U$; and
(c)~$|U| \geq 2$. 
If $U_1$ and $U_2$ are ECs with $U_1 \cap U_2 
\neq \emptyset$, then $U_1 \cup U_2$ is an EC\@.
A \emph{maximal end-component (MEC)} is an EC that is maximal wrt set inclusion.
Every vertex of $V$ belongs to \emph{at most} one MEC and the 
MEC decomposition consists of all MECs of $V$
and all vertices of $V$ that do not belong to \emph{any} MEC\@.
The MEC decomposition problem generalizes 
the strongly connected component (SCC) decomposition of directed graphs 
($V_R=\emptyset$) and closed recurrent sets for Markov chains ($V_1=\emptyset$).

\para{Applications.} 
In verification of probabilistic systems, the classical model is called 
\emph{Markov decision processes (MDPs)}~\cite{Howard},
where there are two types of vertices.
The vertices in $V_1$ are the regular vertices in a graph algorithmic 
setting, and the vertices in $V_R$ represent random vertices.
MDPs are used to model and solve control problems in systems such as stochastic systems~\cite{FV97},
concurrent probabilistic systems~\cite{CY95}, planning problems in artificial intelligence~\cite{Puterman},
and many problems in verification of probabilistic systems~\cite{BK-book}. 
The MEC decomposition problem is a central algorithmic problem in the
verification of probabilistic systems~\cite{CY95,BK-book} and it is a core component
in all leading tools of probabilistic verification~\cite{PRISM,STORM}. 
Some key applications are as follows:
(a)~the almost-sure reachability problem can be solved in linear time given
the MEC decomposition~\cite{CDHL16};
(b)~verification of MDPs wrt $\omega$-regular properties requires 
MEC decomposition~\cite{CY95,CY90,BK-book,CH11};
(c)~algorithmic analysis of MDPs with quantitative objectives as well
as the combination of $\omega$-regular and quantitative objectives 
requires MEC decomposition~\cite{CH-ILC,BBCFK11,CHJS15};
and (d)~applying learning algorithms to verification requires MEC decomposition 
computation~\cite{KPR18,DHKP17}.

\para{Symbolic model and algorithms.}
In verification, a system consists of variables, and a state of the system corresponds to a set of valuations, one for each variable. 
This naturally induces a directed graph: vertices represent states and the directed edges represent
state transitions. 
However, as the transition systems are huge they are usually not explicitly represented during their analysis. 
Instead they are \emph{implicitly represented} using e.g., binary-decision diagrams (BDDs)~\cite{bryant1986graph,bry92}. 
An elegant theoretical model for algorithms that works on this implicit representation, without considering the specifics of the representation and implementation, 
has been developed, called \emph{symbolic algorithms} (see
e.g.~\cite{BurchCMDH90,ClarkeMCH96,Somenzi99,ClarkeGP99,ClarkeGJLV03,CHLOT18,GentiliniPP08,ChatterjeeHJS13}).
A symbolic algorithm is allowed to use the same mathematical, logical, and memory access operations as a regular RAM algorithm, except for the
access to the input graph: 
It is not given  access to the input graph through an adjacency list or adjacency matrix representation but instead 
\emph{only} through two types of \emph{symbolic operations}:
\begin{compactenum}
\item {\emph{One-step operations Pre and Post:}} Each \emph{predecessor \textsc{Pre}} 
	(resp., \emph{successor \textsc{Post}}) operation is given a set $X$ of vertices and
returns the set of vertices~$Y$ with an edge to (resp., edge from) 
some vertex of~$X$.
\item \emph{Basic set operations:} Each basic set operation is given one or two sets of vertices or edges 
and performs a union, intersection, or complement on these sets.
\end{compactenum}
Symbolic operations are more expensive than the non-symbolic operations and 
thus \emph{symbolic time} is defined as the \emph{number of symbolic operations} of a symbolic algorithm. 
One unit of space is defined as one set (not the size of the set) due to the implicit representation as a BDD\@.
We define \emph{symbolic space} of a symbolic algorithm as the maximal \emph{number of sets} stored simultaneously.
Moreover, as the symbolic model is motivated by the compact representation of huge graphs,
we aim for symbolic algorithms that require sub-linear space.

\para{Previous results and main open question.} 
We summarize the previous results and the main open question.
We denote by $|V| = n$ and $|E| = m$ the number of vertices and edges, respectively.
\begin{compactitem}
\item\para{Standard RAM model algorithms.} 
The computation of the MEC (aka controllable recurrent set in early works) 
decomposition problem has been a central problem since the work of~\cite{CY90,CY95,deAlfaroThesis}.
The classical algorithm for this problem requires $O(n)$ SCC decomposition calls,
and the running time is $O(nm)$.
The above bound was improved to (a)~$O(m\sqrt{m})$ in~\cite{CH11} and (b)~$O(n^2)$ in~\cite{CH14}.
While the above algorithms are deterministic, a randomized algorithm with expected almost-linear
$\widetilde{O}(m)$ running time  has been presented in~\cite{CDHS19}.

\item \emph{Symbolic algorithms.}
The symbolic version of the classical algorithm for MEC decomposition requires $O(n)$ 
symbolic SCC computation. 
Given the $O(n)$ symbolic operations SCC computation algorithm 
from~\cite{GPP03}, we obtain an $O(n^2)$ symbolic operations MEC decomposition 
algorithm, which requires $O(\log n)$ symbolic space.
A symbolic algorithm, based on the algorithm of~\cite{CH11}, was presented in~\cite{CHLOT18},
which requires  $O(n \sqrt{m})$ symbolic operations and  $O(\sqrt{m})$ symbolic space.
\end{compactitem}
The classical algorithm from the 1990s with the linear symbolic-operations SCC 
decomposition algorithm from 2003 gives the $O(n^2)$ symbolic operations bound, 
and since then the main open question for the MEC decomposition problem has been 
whether the worst-case $O(n^2)$ symbolic operations bound can be beaten.

\para{Our contributions.}
\begin{enumerate}
\item In this work we answer the open question in the affirmative. 
Our main result presents a symbolic operation and symbolic space trade-off algorithm that
for any $0<\epsilon \leq 1/2$ requires $\widetilde{O}(n^{2-\epsilon})$ symbolic operations and 
$\widetilde{O}(n^{\epsilon})$ symbolic space. 
In particular, our algorithm for $\epsilon=1/2$ requires $\widetilde{O}(n^{1.5})$ symbolic operations 
and $\widetilde{O}(\sqrt{n})$ symbolic space, which improves both the symbolic operations and 
symbolic space of~\cite{CHLOT18}.

\item We also show that our techniques extend beyond MEC computation and is also applicable to 
almost-sure winning (probability-1 winning) of $\omega$-regular objectives for MDPs. 
We consider parity objectives which are cannonical form to express $\omega$-regular objectives.
For parity objectives with $d$ priorities the previous symbolic algorithms require 
$O(d)$ calls to MEC decomposition; thus leading to bounds such as 
(a)~$O(n^2\cdot d)$ symbolic operations and $O(\log n)$ symbolic space;
or (b)~$O(n \sqrt{m} \cdot d)$ symbolic operations and $O(\sqrt{m})$ symbolic space.
In contrast we present an approach that requires $O(\log d)$ calls to MEC decomposition,
and thus our algorithm requires
$\widetilde{O}(n^{2-\epsilon})$ symbolic operations and
$\widetilde{O}(n^{\epsilon})$ symbolic space, for all $0 < \epsilon \leq 1/2$.
Thus we improve the time-space product from $\O(n^2d)$ to $\widetilde{O}(n^2)$.
\end{enumerate}

\para{Technical contribution.} 
Our main technical contributions are as follows:
(1)~We use a separator technique for the decremental SCC algorithm from~\cite{CHILP16}. 
However, while previous MEC decomposition algorithms for the standard RAM model 
(e.g.~\cite{CDHS19}) use ideas from decremental SCC algorithms, data-structures used in 
decremental SCC algorithms of~\cite{CHILP16,BPW19} 
such as Even-Shiloach trees~\cite{EvenS81} have no symbolic representation.
A key novelty of our algorithm is that instead of basing our algorithm on a decremental algorithm we use 
the incremental MEC decomposition algorithm of~\cite{CH11} 
along with the separator technique. 
Moreover, the algorithms for decremental SCC of~\cite{CHILP16,BPW19} are randomized algorithms, 
in contrast, our symbolic algorithm is deterministic.
(2)~Since our algorithm is based on an incremental algorithm approach, we need to support an 
operation of collapsing ECs even though we do not have access to the graph directly (e.g.\ through an adjaceny list representation), but only have access to the graph through symbolic operations.
(3)~All MEC algorithms in the classic model first decompose the graph into its SCCs and then run on each SCC\@. However,
to achieve sub-linear space we cannot store all SCCs, and we show that our algorithm
has a tail-recursive property that can be utilized to achieve sub-linear space.
With the combination of the above ideas, we beat the long-standing $O(n^2)$ symbolic operations barrier
for the MEC decomposition problem, along with sub-linear symbolic space.

\para{Implications.}
Given that MEC decomposition is a central algorithmic problem for MDPs, our result has 
several implications.
The two most notable examples in probabilistic verification are:
(a)~almost-sure reachability objectives in MDPs can be solved with 
$\widetilde{O}(n^{1.5})$ symbolic operations, improving the previous known
$O(n \sqrt{m})$ symbolic operations bound; and
(b)~almost-sure winning sets for canonical $\omega$-regular objectives such as
parity and Rabin objectives with $d$-colors can be solved with 
$O(d)$ calls to the MEC decomposition followed by a call to almost-sure 
reachability~\cite{deAlfaroThesis,ChaThesis}, and hence our result gives  
an $\widetilde{O}(d n^{1.5})$ symbolic operations bound 
improving the previous known $O(d n \sqrt{m})$ symbolic operations 
bound.

\section{Preliminaries}	

\para{Markov Decision Processes (MDPs).}
A \emph{Markov Decision Process} $P= ((V,E), \ls V_1, V_R \rs,
\delta)$ 
consists of a finite set of vertices $V$ partitioned into 
the player-1 vertices $V_1$ and the random vertices $V_R$, 
a finite set of edges $E \subseteq (V \times V)$, 
and a probabilistic transition function $\delta$.  
We define $|V| = n$ to be the number of vertices and $|E| = m$ to be the
number of edges.

The probabilistic transition function $\delta$ maps every random
vertex in $V_R$ to an element of $\D(V)$, where $\D(V)$ is the set of
probability distributions over the set of vertices $V$. A random vertex $v$ has
an edge to a vertex $w \in V$, i.e.\ $(v,w) \in E$ iff
$\delta(v)[w] > 0$. An edge $e = (u,v)$ is a \emph{random edge} if $u \in V_R$ otherwise it is a \emph{player-1	edge}.
W.l.o.g., we assume $\delta(v)$ to be the uniform distribution over vertices $u$ with $(v,u) \in E$.
We follow the common technical assumption that vertices in the MDP do not have self-loops and 
that every vertex has an outgoing edge.
Given a set $X$ of vertices we define $P[X] = (X,(X\times X) \cap E,\langle X \cap V_1, X \cap V_r \rangle,
\delta')$ where $\delta'$ is again the uniform distribution over vertices $u$ with $(v,u) \in (X \times
X) \cap E$. 

\para{Graphs.}
Graphs are a special case of MDPs with $V_R = \emptyset$.
The sets $\In{v,E}$, $\Out{v,E}$ describe the sets
of predecessors and successors of a vertex $v$.
More formally, $\In{v,E}$ is defined as the set $\{w \in V \mid (w,v) \in E \}$
and $\Out{v,E} = \{ w \in V \mid (v,w) \in E \}$. When $E$ is clear from the context we sometimes write
$\In{v}$ and $\Out{v}$.
The diameter $\Diam{S}$ of a set of vertices $S \subseteq V$ is defined as the largest finite
distance in the graph induced by $S$, which coincides with the usual graph-theoretic definition on
strongly connected graphs. 

\para{Maximal End-Component (MEC) Decomposition.}
An \emph{end-component} (EC) is a set of vertices $X \subseteq V$
s.t. (1)~the subgraph induced by $X$ is strongly connected (i.e., $(X,E \cap (X \times X))$ is
strongly connected) and (2)~all random vertices in $X$ have their outgoing edges in $X$. More
formally, for all $v \in X \cap V_R$ and all $(v,u) \in E$ we have $u \in X$.
In a graph, if (1) holds for a set of vertices $X \subseteq V$
we call the set $X$ a strongly connected subgraph (SCS). 
An end-component or a SCS, is \emph{trivial} if it only
contains a single vertex with no edges. All other end-components and SCSs respectively, are \emph{non-trivial}.
A \emph{maximal end-component} (MEC) is an end-component
which is maximal under set inclusion. 
MECs generalize strongly connected components (SCCs) in graphs (where no random vertices exist) and in a MEC $X$, player~1 can almost-surely reach all vertices $u
\in X$ from every vertex $v \in X$. 
The MEC decomposition of an MDP is the partition of the vertex set into MECs and the set of vertices which do
not belong to any MEC, e.g., a random vertex with no incoming edges.

\subsection{Symbolic Model of Computation}
In the set-based symbolic model, the MDP is 
stored by implicitly represented sets of vertices and edges.
Vertices and edges of the MDP are not accessed explicitly but with
set-based symbolic operations. The resources in the
symbolic model of computation are characterized by the number of set-based
symbolic operations and set-based space. 

\para{Set-Based Symbolic Operations.}
A set-based symbolic algorithm can use the same mathematical, logical and memory
access operations as a regular RAM algorithm, except for the access to the graph.
An input MDP $P= ((V,E), \ls V_1, V_R \rs, \delta)$ can be accessed only by the following types of operations:
\begin{enumerate}
	\item Two sets of vertices or edges can be combined with \emph{basic set operations}:
		$\cup,\cap,\subseteq,\setminus$, $\times$ and $=$.
	\item The \emph{one-step operation} to obtain the predecessors/successors of
		the vertices of $S$ w.r.t.\@ the edge set $E$. In particular, we define the predecessor and successor
		operation over a specific edge set $E$:
		\begin{align*}
			\Pre{E}{S} &= \{ v \in V \mid \Out{v,E} \cap S \neq \emptyset \} \text{ and }\\
			\Post{E}{S} &= \{ v \in V \mid \In{v,E} \cap S \neq \emptyset \}
		\end{align*}
	\item The $\Pick{S}$ operation which returns an arbitrary vertex $v \in S$ and the $|S|$
		operation which returns the cardinality of $|S|$.
		
\end{enumerate}

Notice that the subscript $E$ of the one-step operations is often omitted when the edge set is clear from the context.
As our algorithm deals with different edges sets we make the edge set explicit as a subscript.

\para{Set-based Symbolic Space.}
The basic unit of space for a set-based symbolic algorithm for MDPs are
sets~\cite{BrowneCJLM97,chatterjee2017symbolic}.  For example, a set can
be represented symbolically as one
BDD~\cite{bryant1986graph,bry92,BurchCMDH90,ClarkeMCH96,Somenzi99,ClarkeGP99,ClarkeGJLV03,GentiliniPP08,ChatterjeeHJS13,CHLOT18}
and each such set is considered as unit space.  Consider for example an MDP 
whose state-space consists of valuations of $N$-boolean variables. The
set of all vertices is simply represented as a true BDD\@. Similarly, the set of all
vertices where the $k$th bit is false is represented by a BDD which depending
on the value of the $k$th bit chooses true or false. Again, this set can be
represented as a constant size BDD\@. 
Thus, even large sets of vertices can sometimes be represented as constant-size BDDs.
In general, the size of the smallest BDD representing a set is notoriously hard to
determine and depends on the variable reordering~\cite{ClarkeGP99}.
To obtain a clean theoretical model for the algorithmic analysis,
each set is represented as a unit data structure and
requires unit space.  Thus, for the \emph{space requirements} of a symbolic
algorithm, we count the maximal number of sets the algorithm stores
simultaneously and denote it as the {\em symbolic space}.

\section{Algorithmic Tools}
In this section, we present various algorithmic tools that we use in our algorithm. 

\subsection{Symbolic SCCs Algorithm}\label{ss:symbolicscc}
In~\cite{GPP03} and~\cite{ChatterjeeDHL18} symbolic algorithms and lower bounds for computing the
SCC-decomposition are presented.
Let $D$ be the diameter of $G$ and let $D_C$ be the diameter of 
the SCC $C$. The set $\SCCs{G}$ is a family of sets, where each set contains one SCC of $G$.  
The upper bounds are summarized in Theorem~\ref{thm:SCCs}.
\begin{theorem}[\cite{GPP03,ChatterjeeDHL18}]\label{thm:SCCs}
	The SCCs of a graph $G$ can be computed in
	$\Theta(\min(n, |SCCs(G)|\cdot D, \sum_{C \in \SCCs{G}} (D_C + 1)))$ symbolic operations and $\O(1)$ symbolic space. 
\end{theorem}
Note that the SCC algorithm of~\cite{GPP03} can be easily adapted
to accept a starting vertex which specifies the 
SCC computed first.
Also, SCCs are output when they are detected by the algorithm
and can be processed before the remaining SCCs of the graph are computed. 
We write $\SCCFind{P}{V',s}$ when we refer to the above described algorithm for computing 
the SCCs of the subgraph with vertices $V' \subseteq V$ and using $s \in V$ as the starting vertex for the algorithm 
in the MDP $P = (V,E,\langle V_1, V_R \rangle, \delta)$. 
When we consider SCCs in an MDP we take \emph{all} edges into account 
and, i.e., also the random edges.

\subsection{Random Attractors}
Given a set of vertices $T$ in an MDP $P$, the random attractor $\Attr{R}{P}{T}$ is a set of vertices
consisting of (1) $T$, (2) random vertices with an edge to some vertex in $\Attr{R}{P}{T}$, (3) player-1 vertices
with all outgoing edges in $\Attr{R}{P}{T}$.  
Formally, given an MDP $P = (V,E, \langle V_1, V_R \rangle, \delta)$, let $T$ be a set of vertices.
The random attractor $\Attr{R}{P}{T}$ of $T$ is defined as: $\Attr{R}{P}{T} = \bigcup_{i\geq 0} A_i$ where $A_0 = T$ and $A_{i+1} = A_i \cup
(\Pre{E}{A_i} \setminus (V_1 \cap \Pre{E}{V \setminus A_i}))$ for all $i>0$. We sometimes refer to $A_i$ as the $i$-th level of the attractor.

\begin{lemma}[\cite{CHLOT18}]\label{lem:attrRunning}
	The random attractor $\Attr{R}{P}{T}$ can be computed with at most $O(|\Attr{R}{P}{T}\setminus{T}|+1)$ many symbolic
	operations.
\end{lemma}
The lemma below establishes that the random attractor of random vertices with edges out of a strongly connected set
is not included in any end-component and that it can be removed without affecting the ECs of the remaining graph. Hence, we use the lemma to identify vertices that do not belong to any EC\@.
The proof is analogous to the proof in Lemma~2.1 of~\cite{CH11}.

\begin{lemma}[\cite{CH11}]\label{lem:attr_remove}
	Let $P = (V,E,\ls V_1, V_R \rs, \delta)$ be an MDP\@.
	Let $C$ be a strongly connected subset of $V$.
	Let $U = \Set{v \in C \cap V_R \mid \Out{v} \cap (V \setminus C) \neq \emptyset}$ 
	be the random vertices in $C$ with edges out of $C$. 
	Let $Z = \Attr{R}{P}{U} \cap C$. 
	For all non-trivial ECs $X \subseteq C$ in $P$ we have $Z \cap X = \emptyset$.
\end{lemma}

\subsection{Separators}~\label{ss:separator}
Given a strongly connected set of vertices $X$ with $|X| = n$, a separator is a non-empty set $T \subseteq X$ such that  
the size of every SCC in the graph induced by $X \setminus T$ is small. 
More formally, we call $T \subseteq X$ a $q$-separator if each SCC in the subgraph induced by $X \setminus T$ has at most $n
-q \cdot |T|$ vertices. For example, if we compute a $q = \sqrt{n}$-separator, the SCCs in the
subgraph induced by $X \setminus T$ have size $\leq n - \sqrt{n}$ as $|T|$ is non-empty. In
\cite{CHILP16}, the authors
present an algorithm which computes a $q$-separator when the diameter of $X$ is large. We briefly sketch the symbolic 
version of this algorithm. 

The procedure $\Call{Separator}{X,r,\gamma}$ computes a $q = \lfloor \gamma/(2 \log n) \rfloor$-separator $T$
when $X$ has diameter at least $\gamma$:

\begin{enumerate}
	\item Let $q \gets \lfloor \gamma/(2 \log n) \rfloor$.
	\item Try to compute a BFS tree $K$ of either $X$ or the reversed graph of $X$ of depth at least $\gamma$ 
with an arbitrary vertex $r \in X$ as root. In the symbolic algorithm, we build the BFS trees
with $\Pre{}{\cdot}$ and $\Post{}{\cdot}$ operations.
	\item If the BFS trees of both $X$ and the reversed graph $X$ with root $r$ have less than $\gamma$ levels then return $\emptyset$. Note that the
		diameter of $X$ is then $\leq 2\gamma$.
	\item Let layer $L_i$ of $K$ be the set of vertices $L_i \subseteq X$ with distance $i$ from $r$.
		Due to~\cite[Lemma 6]{CHILP16}, there is a certain layer $L_i$ of the BFS tree $K$ which is a $q$-separator. 
		Intuitively, they argue that removing the layer $L_i$ of $K$ separates $X$ into two parts: (a)
		$\bigcup_{j<i} L_j$ and (b) $\bigcup_{j>i} L_j$
		which cannot be strongly connected anymore due to the fact that $K$ is a BFS tree.
		They show that one can always efficiently find a layer $L_i$ such that both part (a) and part (b) are
		small.
	\item We efficiently find the layer $L_i$ while building $K$.
\end{enumerate}

The detailed symbolic implementation of $\Call{Separator}{X,\gamma}$ is illustrated in Figure~\ref{alg:separator}.

\begin{figure}
 \begin{algorithmic}
	 \Procedure{Separator}{$X,\gamma$}
		 \State $q \gets \lfloor \gamma/(2 \log n) \rfloor;$ $v \gets \Pick{X};$ 
		 \State $i\gets 0;$ $K \gets \Set{v};$ $c \gets 0;$ $L \gets \emptyset;$ $R \gets \emptyset;$
		\While{$(\Post{E}{K} \cap X) \not\subseteq K$}\label{alg:separator:while1}
			\If{$q \leq i \leq \gamma/2$} 
			\State $Z \gets \Post{E}{K}\setminus K \cap X$;
				\lIf{$L = \emptyset$ and $|Z| \leq 2^{i/q-1} $} $L \gets Z$
			\EndIf	
			\If{$\gamma/2 \leq i \leq \gamma - q$ } 
				$Z \gets \Post{E}{K}\setminus K \cap X$;
				\lIf{$|Z| \leq 2^{(\gamma-i)/q-1}$} $R \gets Z$;
			\EndIf	
			\lIf{$i \leq \gamma/2$} $c\gets c+ |\Post{E}{K}\setminus K \cap X|$
			\State $K \gets K \cup (\Post{E}{K} \cap X)$;
			\State $i \gets i +1$;
		\EndWhile	
		\If{$i < \gamma$}
			\State $i\gets 0,\ K \gets \Set{v},\ c\gets 0, L \gets \emptyset, R \gets \emptyset$;
			\While{$i \leq \gamma$ and $(\Pre{E}{K} \cap X) \not\subseteq K$}\label{alg:separator:while2}
				\If{$q \leq i \leq \gamma/2$}
					$Z \gets \Pre{E}{K}\setminus K \cap X$;
					\lIf{$L = \emptyset$ and $|Z| \leq 2^{i/q-1} $} $L \gets Z$;
				\EndIf	
				\If{$\gamma/2 \leq i \leq \gamma - q$} 
					\State $Z \gets \Pre{E}{K}\setminus K \cap X$;
					\lIf{$|Z| \leq 2^{(\gamma-i)/q-1} $} $R \gets Z$;
				\EndIf	
				
				\lIf{$i \leq \gamma/2$}	$c\gets c+ |\Pre{E}{K}\setminus K \cap X|$
				
				\State $K \gets K \cup (\Pre{E}{K} \cap X)$;
				\State $i \gets i +1$;
			\EndWhile	
		\EndIf
		\If{$i < \gamma$}
			\Return{$\emptyset$;} \Comment{$\Diam{X} < 2\gamma$}
		\EndIf	
		\If{$c< |X|/2$}
			\State \Return{$L$};
		\Else
			\State \Return{$R$};
		\EndIf
	\EndProcedure
 \end{algorithmic}
 \caption{If the strongly connected set $X$ has diameter $\gamma$ compute separator with
	 quality $\lfloor \gamma/(2\log k) \rfloor$.}\label{alg:separator}
\end{figure}

The following lemmas summarize useful properties of $\Call{Separator}{X,r,\gamma}$.
\begin{observation}[\cite{CHILP16}]\label{obs:SepSize}
	A $q$-separator $S$ of a graph $G$ with $|V|=n$ vertices contains at most $\frac{n}{q}$ vertices, 
	i.e., $|S| <  \frac{n}{q}$.
\end{observation}
 \begin{proof}
 	Assume for contradiction that $|S| \geq n/q$. Trivially, given a $q$-separator $S$, 
 	each SCC $C_i \in \SCCs{G}$ contains at most $n - q |S|$ vertices and at least one vertex.
 	But then $C_i \leq n - q |S| \leq n - q \frac{n}{q} \leq 0$, a contradiction.
 \end{proof}

\begin{lemma}[\cite{CHILP16}]\label{lem:sepqual}
	Let $X$ be a strongly connected set of vertices with $|X| = k$, let $r \in
	X$, and	let $\gamma$ be an integer such that $q=\lfloor \gamma/(2\log k) \rfloor \geq 1$. Then
	$\Call{Separator}{X,r,\gamma}$ 
	computes a $\lfloor \gamma / (2\log k) \rfloor$-separator
	if there exists a vertex $v \in X$ where the distance
	between $r$ and $v$ is at least $\gamma$. If no
	such vertex $v$ exists, then $\Call{Separator}{X,r,\gamma}$ returns the empty set.
\end{lemma}

The following lemma bounds the symbolic resources of the algorithm. 

\begin{lemma}\label{lem:runspacesep}
	$\Call{Separator}{X,r,\gamma}$ runs in $O(|X|)$ symbolic operations and uses $O(1)$ symbolic space.
\end{lemma}
 \begin{proof}
 	The bound on the symbolic operations of the while loop at Line~\ref{alg:separator:while1} is clearly
 	in $O(|X|)$ because we perform $\Post{E}{K} \cap X$ operations until the set $K$ fully contains
 	$X$ or $\Post{E}{K} \cap X$ is fully contained in $K$. Note that we only perform a constant amount
 	of symbolic work in the body of the while-loop.
 	As each $\Post{E}{K} \cap X$ operation adds at least one vertex to $X$ until termination, we perform
 	$O(|X|)$ many operations in total. 
 	A similar argument holds for the while-loop at Line~\ref{alg:separator:while2}.
 	Note that we use a constant amount of sets in Algorithm~\ref{alg:separator} which implies 
 	$O(1)$ symbolic space.
 \end{proof}

\section{Symbolic MEC decomposition}
In this section, we first define how we collapse end-components.
Then we present the algorithm for the symbolic MEC decomposition.

\subsection{Collapsing End-components}
A key concept in our algorithm is to collapse a detected EC $X \subseteq V$ of an MDP $P =
(V,E,\langle V_1,V_R \rangle, \delta)$ to a single vertex $v \in X$ in order to speed up the
computation of end-components that contain $X$.
Notice that we do not have access to the graph directly, but only have access to the graph through symbolic operations.
\begin{figure}
	\begin{algorithmic}[1]
		\Procedure{CollapseEC}{$X,P=(V,E,\langle V_1,V_R\rangle, \delta)$}
		\If{$X \cap V_1 \neq \emptyset$}
		\State $v \gets \Pick{X \cap V_1}$;
		\Else
		\State $v \gets \Pick{X}$;
		\State $V_1 \gets V_1 \cup \{v\}$; $V_R \gets V_R \setminus \{v\}$;
		\EndIf	
		\State //~add incoming and outgoing edges of $X$ to $v$
		\State $E \gets E \cup (((\Pre{E}{X}\setminus X) \times \Set{v}) \cup (\Set{v} \times (\Post{E}{X}\setminus X)))$; 
		\State //~remove edges of $X\setminus \{v\}$ and remove $X\setminus \{v\}$ from vertex sets
		\State $E \gets E \setminus ((V  \times (X\setminus v)) \cup ((X\setminus v) \times V))$; 
		\State $V \gets V \setminus (X \setminus \{v\})$; $V_1 \gets V_1 \cap V$; $V_R \gets V_R \cap V$;
		\EndProcedure
	\end{algorithmic}
	\caption{Collapsing an end-component}\label{alg:CollapseEC}
\end{figure}

We define the \emph{collapsing} (see $\Call{CollapseEC}{X,P}$, Figure~\ref{alg:CollapseEC}) of an EC $X$ as picking a
vertex $v \in X$ (player-1 if possible) which \emph{represents} $X$, directing all the incoming edges $X$ to $v$, directing the outgoing edges of $X$
from $v$ and removing all edges to and from vertices in $X\setminus \{v\}$. The procedure removes the
vertices in $X\setminus\{v\}$ from the MDP $P$. For vertices not in $X$ we have that they are in a non-trivial MEC in the modified MDP 
iff they are in a non-trivial MEC in the original MDP (before collapsing).
We denote with $\AEC{P}$ the set of all ECs in $P$ (including nontrivial ECs).
The following lemma summarizes the property as observed in~\cite{CH11}.

\begin{lemma}\label{lem:collapse}
    For MDP $P = (V,E,\ls V_1, V_R \rs, \delta)$ with $X \in \AEC{P}$
	and the MDP $P'$ that results from collapsing $X$ to $v\in X$ with $\Call{CollapseEC}{X,P}$ we have: 
    (a) for $D \subseteq V \setminus X$ we have $D \in \AEC{P}$ iff $D \in \AEC{P'}$; 
    (b) for $D \in \AEC{P}$ with $D \cap X \not= \emptyset$ we have $(D \setminus X) \cup \{v\} \in
	\AEC{P'}$; 
	and (c) for $D \in \AEC{P'}$ with $v \in D$ we have $D \cup X \in \AEC{P}$.
\end{lemma}
It is a straightforward consequence of Lemma~\ref{lem:collapse}~(a) that ECs in $P'$ which do not include vertex
$v$ are an EC in $P$.

\begin{figure}
	\begin{algorithmic}[1]
		\Procedure{SymbolicMec}{$P' = (V',E',\ls V'_1, V'_R \rs, \delta'),\gamma$}
		\State $M \gets \emptyset, \textsc{MECs} \gets \emptyset, P \gets P'$;

		\While{$C \gets \SCCFind{P'}{V', \emptyset}$}\label{alg:metasymmec:while1} \Comment{stage 1}
		\State $M \gets M \cup \Call{SymMec}{C, \gamma,P}$\Comment{uses P instead of P'}\label{alg:metasymmec:rec1}
		\EndWhile	

		\While{$C \gets \SCCFind{P'}{M, \emptyset}$}\Comment{stage 2}\label{alg:metasymmec:while2}
			\State $\textsc{MECs} \gets \textsc{MECs} \cup \{C\}$\label{alg:metasymmec:addSCC}
		\EndWhile	
		\State \Return{$\textsc{MECs}$}\label{alg:metasymmec:returnMECS}
		\EndProcedure
	\end{algorithmic}
	\caption{Computing the MEC decomposition of an MDP $P'$.}\label{alg:metasymmec}
\end{figure}

\subsection{Algorithm Description}
The procedure \Call{SymbolicMec}{$P',\gamma$} has two stages.
In the first stage, we compute the set $M\subseteq V$ of all vertices that are in a non-trival MEC
and then, in the second stage, we use this set $M$
to compute the MECs with an SCC algorithm (see Figure~\ref{alg:metasymmec}).

In the first stage, we iteratively compute the SCCs 
$C_1, \dots, C_\ell$ of $P'$ and immediately call $\Call{SymMec}{C_i,\cdot,P}$ (cf. Figure~\ref{alg:symmec}) 
to compute the vertices $M_i$ of $C_i$ that are in a non-trival MEC\@. 
The set $M$ is the union over these sets, i.e., $M = \bigcup_{1 \leq i \leq \ell} M_i$.
$\Call{SymMec}{\cdot,P}$ applies collapsing operations and thus modifies the edge set of the MDP\@.
We hand a copy $P$ of the original MDP $P'$ to $\Call{SymMec}{\cdot,P}$.
Finally, to obtain all non-trivial MECs in $P'$ we restrict the graph to the vertices in $M$ and compute the SCCs which correspond to the MECs. Note that trivial MECs are player-1 vertices that are not contained in any non-trivial MEC and thus can be simply computed
by iterating over the vertices of $V_1 \setminus M$.

In the following, we focus on the procedure $\Call{SymMec}{\cdot,P}$ which is the core of our algorithm.
To this end, we introduce the operation $\ROut{S} = \Pre{E}{V \setminus S} \cap (S \cap V_{R})$, 
that computes the set of random vertices in $S$ with edges to $V \setminus S$. 

\para{The procedure $\Call{SymMec}{S,\gamma,P}$} works in a recursive fashion: The input is a set $S$ of strongly connected
vertices and a parameter $\gamma$ for separator computations that is fixed over all recursive calls (cf.\ Figure~\ref{alg:symmec}).
The main idea is to compute a separator to divide the original graph into smaller SCCs, recursively compute the vertices
which are in MECs of the smaller SCCs, and 
then compute the MECs of the original graph by incrementally adding the vertices of the separator
back into the recursively computed MEC decomposition.
In each recursive call, we first check whether the given set $S$ is larger than one 
(if not it cannot be a non-trivial MEC) and then 
if $S$ is a non-trivial EC by checking if $\ROut{S} = \emptyset$.
If $S$ is a non-trivial EC we collapse $S$ and return $M$. 
If $\ROut{S} \neq \emptyset$, $S$ is not an EC\@ but it may contain nontrivial ECs of $C$\@. 
We then try to compute a balanced separator $T$ of $S$ which is nonempty
if the diameter of $S$ is large enough ($\geq 2\gamma$). 
We further distinguish between the two cases: 
In the first case, we succeed to compute the balanced separator and
we recurse on the strongly connected components in $\SCCs{S\setminus\Attr{R}{P[S]}{T}}$.\vspace{0.05cm}
After computing and collapsing the ECs in $\SCCs{S\setminus\Attr{R}{P[S]}{T}}$
we start from $S \setminus T$ and incrementally add vertices of $T$ and compute the new ECs until all vertices of $T$ have been added.
In each incremental step, we add one vertex $v$ and find the SCC of $v$ in current set. 
If the random attractor of $\ROut{S'}$ (note that this computation now considers all vertices in $P$) 
does not contain the whole SCC $S'$ we are able to prove that we can identify a new EC of $C$. 
Otherwise, the vertex $v$ does not create a new EC in $S \setminus T$.
In the second case where we fail to compute the balanced separator we know
that the diameter of $S$ is small ($<2\gamma$). We remove the
random attractor $X$ of $\ROut{S}$ (this set cannot contain ECs due to Lemma~\ref{lem:attr_remove}), 
recompute the strongly connected components in the set $S \setminus X$ and recurse on 
each of them one after the other.
\begin{figure*}
	\begin{center}
		\includegraphics[width=\linewidth]{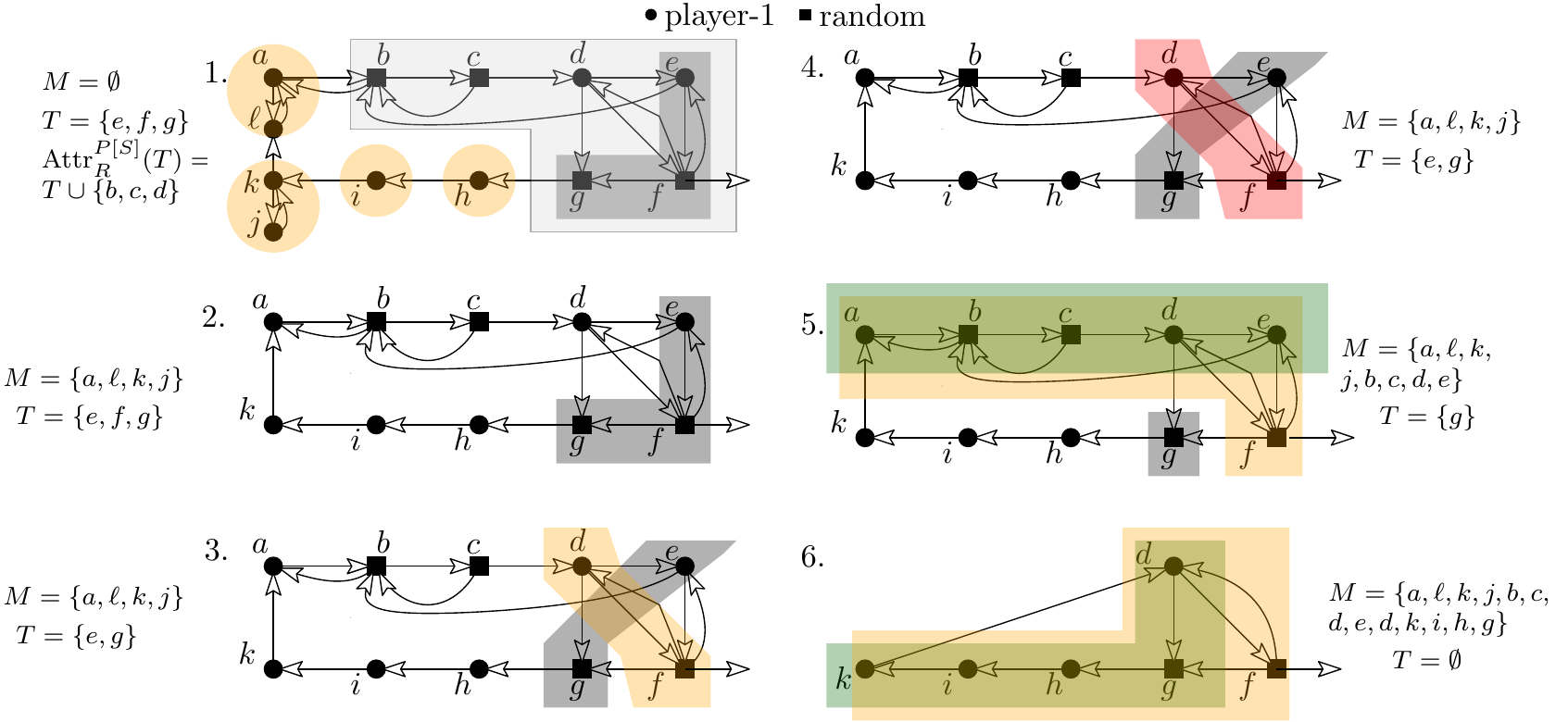}
	\end{center}
	\caption{Example Execution of $\protect\Call{SymMec}{\cdot}$ on a strongly connected set of
		vertices with a random edge
		leaving it.}\label{fig:symmec}
\end{figure*}

Figure~\ref{fig:symmec} illustrates a call to 
$\Call{SymMec}{}$ with a strongly connected set of vertices $S=\{a,b,c,d,e,f,g,h,i,j,k,\ell\}$.
In the first step, we compute that $\ROut{S}=\{f\}$ is nonempty, i.e., $S$ cannot be a MEC\@.
We then successfully compute a separator $T$ of $S$ and recurse on the SCCs of $S \setminus
\Attr{R}{P[S]}{T}=\{\{a,\ell\},\{k,j\},\{i\},\{h\}\}\}$ (orange circles).
In the second step, we collapse two nontrivial MECs identified by the recursion and add them to $M$. 
Next, we execute the incremental EC detection for all vertices in $T$:
In step three, we pick $f \in T$, remove it
from $T$ and compute its SCC $S'= \{d,f\}$ in $S\setminus T$. In step four, we compute the random attractor of
$\ROut{S'}=\{f\}$ (in the MDP $P[S']$) which is equal to $S'$ and,
thus, $M$ remains unchanged. 
In step five, we remove $e$ from $T$ and compute its SCC $S' = \{a,b,c,d,e,f\}$ in $S\setminus T$. In this iteration, the random
attractor $Z = \{f\}$ of $\ROut{S'}=\{f\}$ in $P[S']$ is not equal to $S'$ and we add the SCC of
$e$ in $S' \setminus Z$, i.e.,$\{a,b,c,d,e\}$ to the set of nontrivial MECs $M$ and collapse it. 
In step six, we analyze the SCC of the last vertex in the separator $T$, i.e., $g$
and identify an EC consisting of the vertices $\{d,g,h,i,k\}$ which we add to $M$.
Since $T$ is empty $\Call{SymMec}{S,\gamma,P}$ returns $M$. 

\begin{remarkb}
	Note that our symbolic algorithm does not require randomization which is in contrast to the best
	known MEC decomposition algorithms for the standard RAM model~\cite{BPW19}.
	The latter algorithms rely on decremental SCCs algorithms that are randomized as they maintain
	ES-trees~\cite{EvenS81} from randomly chosen centers. Instead, our symbolic algorithm relies on a
	deterministic incremental approach. 
\end{remarkb}

\begin{figure}[ht]
	\begin{algorithmic}[1]
		\Procedure{SymMec}{$S,\gamma,P$}
			\State $M \gets \emptyset$; \label{alg:symmec:initM}
			\lIf{$|S| \leq 1$} \Return{$\emptyset$}; \label{alg:symmec:ret0}
			\If{$\ROut{S} = \emptyset$}\Comment{check if $S$ is an EC}\label{alg:symmec:if1}
				\State $\Call{CollapseEC}{S,P}$; \label{alg:symmec:collapse1}
				\State \Return{$S$}; \label{alg:symmec:ret1}
			\EndIf	
			\State $(T = \Set{v_1, \dots, v_t}) \gets \Call{Separator}{S,\Pick{S},\gamma}$; \label{alg:symmec:sep}
			\If{$T \neq \emptyset$} \Comment{$T \neq \emptyset$ if $\Diam{S} \ge 2\gamma$}
			\State $A \gets \Attr{R}{P[S]}{T}$;\label{alg:symmec:attrSep}
				\While{$S_j \gets \SCCFind{P}{S \setminus A, \emptyset}$}\label{alg:symmec:while1}
				\State $M \gets M \cup \Call{SymMec}{S_j, \gamma,P}$; \label{alg:symmec:rec1}
				\EndWhile	
				\While{$T \neq \emptyset$}\label{alg:symmec:while2}\Comment{incremental EC detection for $v \in T$}
					\State $v \gets \Pick{T}$;\label{alg:symmec:pick}
					\State $T \gets T \setminus \{v\}$;\label{alg:symmec:removevfromT}
					\State $S' \gets \SCCFind{P}{S \setminus T,\Set{v}}$; \label{alg:symmec:computescc2}
					\lIf{$|S'| = 1$} \Continue; \Comment{EC is trivial}\label{alg:symmec:vtrivial}
					\State $Z \gets \Attr{R}{P[S']}{\ROut{S'}}$;\label{alg:symmec:attrIncr}
					\State $S' \gets S' \setminus Z$;\label{alg:symmec:snew}
					\If{$S' \neq \emptyset$}\Comment{$S'$ contains a non-trivial EC }
						\State $U \gets \SCCFind{P}{S',\Set{v}}$; \label{alg:symmec:bottomSCC}
						\State $\Call{CollapseEC}{U,P}$; \label{alg:symmec:collapse2}
						\State $M \gets M \cup U$; \label{alg:symmec:addMEC}
					\EndIf	
				\EndWhile	
				\State \Return{$M$};\label{alg:symmec:ret2}
				\Else \Comment{$T = \emptyset$ and thus  $\Diam{S} < 2\gamma$}
				\State $X \gets \Attr{R}{P[S]}{\ROut{S}}$;\label{alg:symmec:attrB}
				\While{$S_j \gets \SCCFind{P}{S\setminus X, \emptyset}$}\label{alg:symmec:while3}
				\State $M \gets M \cup \Call{SymMec}{S_j, \gamma,P};$\label{alg:symmec:rec2}
				\EndWhile	
				\State \Return{$M$};\label{alg:symmec:ret3}
			\EndIf	
		\EndProcedure
	\end{algorithmic}
	\caption{Recursively compute all vertices in nontrivial end-components in a strongly connected set of vertices.}\label{alg:symmec}
\end{figure}

\subsection{Correctness}

We prove the correctness of $\Call{SymMec}{S,\gamma,P}$ first which will then imply the correctness
of $\Call{SymbolicMec}{P',\gamma}$.
To this end, consider an MDP $P = (V,E,\ls V_1, V_R \rs, \delta)$ and an SCC $C$ of $P$.
The non-trivial end-components of a subset $S$ of $V$ are given by $\EC{S}$ and 
the set of vertices that appear in non-trivial (maximal) ECs of $C$ are denoted by 
$M_C=\bigcup_{Q \in \EC{C}} Q$.
We observe that in all calls to $\Call{SymMec}{S,\gamma,P}$ (see Figure~\ref{alg:symmec})
$S$ is strongly connected.
\begin{lemma}\label{lem:corr:0}
	For a strongly connected set $S \subseteq C$ we have 
	that in the computation of $\Call{SymMec}{S,\gamma,P}$ 
	for all the calls $\Call{SymMec}{S_j,\gamma,P}$ the set 
	$S_j$ is strongly connected and $S_j \subseteq S$.
\end{lemma}
 \begin{proof}
     If $|S| =1$ or $\ROut{S} = \emptyset$ there are no recursive calls and the statement is true.
     Now consider the case where $T \not= \emptyset$.
     At Line~\ref{alg:symmec:rec1}, we consider $S_j$ which is an SCC in 
 	the graph $S \setminus (T \cup A)$, obviously a subset of $S$ and strongly connected.
 	Now consider the case where $T = \emptyset$.
 	At Line~\ref{alg:symmec:rec2} we consider $S_j$ which is an SCC of 	$S \setminus X$, obviously a subset of $S$ and strongly connected. 
 \end{proof}

Let $M$ be the set returned by $\Call{SymMec}{S,\gamma,P}$.
Note that the ultimate goal of $\Call{SymMec}{S,\gamma,P}$ is to compute $M_C$ for a given SCC $C$ of $P'$.
The next lemma shows that every call to $\Call{SymMec}{S,\gamma,P}$ with a strongly connected set $S$
returns a set containing $M_S=\bigcup_{Q \in \EC{S}} Q$ and collapses all ECs $Q \in \EC{S}$ in $P$.

\begin{lemma}\label{lem:allECInM}  
	Let  $S$ be a strongly connected set of vertices and $M$ be the set returned by
	$\Call{SymMec}{S,\gamma,P}$, then for all non-trivial ECs $Q \in \EC{S}$ in $P'$
	we have $Q \subseteq M$ and $Q$ is collapsed in $P$.
\end{lemma}
\begin{proof}
	Let  $Q \in \EC{C}$ with $Q \subseteq S$.
	We prove the statement by induction on the size of $S$.
	For the base case, i.e., if $S$ is empty or contains only one vertex, 
	$\Call{SymMec}{S,\gamma,P}$ returns the empty set at Line~\ref{alg:symmec:ret0} and thus the condition holds.

	For the inductive step, let $|S| > 1$.
	If $Q = S$, we detect it at Line~\ref{alg:symmec:if1} 
	and return $S$ at Line~\ref{alg:symmec:ret1}. 
	Then we collapse the set of vertices at Line~\ref{alg:symmec:collapse1}.
	Thus, $M$ contains $Q$ and the claim holds.
	We distinguish two cases concerning $T$ at Line~\ref{alg:symmec:sep}:
	First, if $T$ is non-empty then the random attractor $A$ of $T$ at Line~\ref{alg:symmec:attrSep}
	is non-empty.
	Thus the SCCs of $S \setminus A$ are strictly smaller than $S$. 
	Thus, we can use the induction hypothesis to stipulate that if $Q$ is contained 
	in one of the SCCs computed at Line~\ref{alg:symmec:while1} 
	we have $Q \subseteq M$ after the while-loop at Lines~\ref{alg:symmec:while1}-\ref{alg:symmec:rec1}. 
	Additionally, by the induction hypothesis, they are collapsed. 
	Also, if $Q$ contains a subset $Q' \subset Q$ that is a non-trivial EC
	and contained in one of the SCCs of $S \setminus A$,
	we have that $Q' \subseteq M$ after the while-loop at Lines~\ref{alg:symmec:while1}-\ref{alg:symmec:rec1}.
	Again, by induction hypothesis, $Q'$ is collapsed.
	We proceed by proving that the separator $T$ must contain one of the vertices of $Q$.
	\begin{claim}\label{claim:TcontainsQ}
		For any non-trivial EC $Q \in \EC{S}$ which is not fully contained
		in one of the SCCs computed at Line~\ref{alg:symmec:while1} we have $T \cap Q \not= \emptyset$.
	\end{claim}
	\begin{proof}
		Let $Q$ be an arbitrary EC in $\EC{S}$.
		Observe that if $Q$ is not fully contained in an SCC of $S \setminus A$ 
		(computed at Line~\ref{alg:symmec:while1}), we have $Q \cap A \neq \emptyset$.
		We now show that $Q \cap T \neq \emptyset$:
		Assume the contrary, i.e., $Q \cap T = \emptyset$. 
		Consequently, $Q$ is a strongly connected set in $S \setminus T$ and not 
		fully contained in an SCC of $S \setminus A$. 
		Also, $Q$ contains a vertex of $A \setminus T$. 
		Let $v \in Q \cap (A \setminus T)$ such that there is no $v' \in Q \cap (A \setminus T)$ which 
		is on a lower level of the attractor $A$.
		If $v \in V_1$, by the definition of the attractor and the above assumption, all its outgoing edges leave the set $Q$, which is in contradiction to $Q$ being strongly connected.
		Thus we have $v \in V_R$ and, by the definition of the attractor,
		$v$ must have an edge to a smaller level and by the above assumption the vertex in the
		smaller level is not in $Q$.
		That is, $Q$ has an outgoing random edge which contradicts the assumption that $Q$ is an EC\@.		
	\end{proof}

	The following claim shows an invariant for the while-loop at Line~\ref{alg:symmec:while2}.
	\begin{claim}\label{claim:non-trivialEC}
		For iteration $j \geq 0$ of the while-loop at Line~\ref{alg:symmec:while2} holds: (a) There are
		no non-trivial ECs $Q \subseteq S \setminus T$ in $P$
		and (b) in each iteration $j$  every nontrivial maximal EC in $S \setminus T$ is collapsed and added to $M$.
	\end{claim}
	\begin{proof}
		The induction base $j=0$ holds 
		for the following reasons: $A \setminus T$ cannot include $Q$ due to Claim~\ref{claim:TcontainsQ}.
		$S \setminus A$ cannot contain $Q$ 
		due to the fact that $Q$ must contain a vertex 
		in $T$ and that we collapsed the ECs contained in SCCs computed at 
		Line~\ref{alg:symmec:while1}.
		For the induction step, assume that there are no non-trivial ECs in $S \setminus T$
		before iteration $j = \ell$.
		When we include one vertex $v$ of $T$ into $S$ 
		at Line~\ref{alg:symmec:pick} there can only be one new non-trivial maximal EC $Q_v$ in $S$, 
		i.e., the one including $v$, as the remaining graph is unchanged. 
		We argue that if such an EC $Q_v$ exists it is collapsed into a vertex and added into $M$ before iteration $\ell+1$ of the while-loop 
		at Line~\ref{alg:symmec:while2}:
		The EC $Q_v$ must be strongly connected, and thus it is contained in the SCC $S'$ of $v$ 
		at Line~\ref{alg:symmec:computescc2}. Note, that if $|S'| = 1$, the new EC is trivial, which 
		concludes the proof. Otherwise, we compute $Z$ at Line~\ref{alg:symmec:attrIncr}, which does not
		contain a vertex of $Q_v$ by Lemma~\ref{lem:attr_remove} and remove it from $S'$. 
		If $S'$ is non-empty, it contains exactly one maximal EC as we argued above. 
		That is the EC $Q_v$ and we thus add $Q_v$ to $M$ and collapse it. 
		As a consequence there is no non-trivial EC in $S \setminus T$ left.
		Thus the claim holds. 
	\end{proof}	
	
	We show that $\Call{SymMec}{S,\gamma,P}$ detects $Q$ 
	if it is not fully contained in one of the SCCs computated at Line~\ref{alg:symmec:while1}.
	By Claim~\ref{claim:TcontainsQ} and because we iteratively remove vertices from $T$
	(Lines~\ref{alg:symmec:pick}--\ref{alg:symmec:removevfromT}) there exists an iteration of the while loop at Line~\ref{alg:symmec:while2} such that $Q
	\subseteq S \setminus T$ holds at Line~\ref{alg:symmec:removevfromT}.
	We show that $Q$ is in $M$ after this iteration:
	Because $Q$ might contain non-trivial ECs $Q'_i$ that were already collapsed due to 
	the recursive call at Line~\ref{alg:symmec:rec1}, or in a prior iteration of the while-loop at Line~\ref{alg:symmec:while2},
	it follows from Lemma~\ref{lem:collapse}(b) that there exists 	
	an EC $Q' \subseteq Q$  in $P$ such that $Q' \supseteq (Q\setminus M)$
	and $Q'$ contains one vertex corresponding to each collapsed sub-ECs $Q'_i$.
	Due to Claim~\ref{claim:non-trivialEC} we collapse $Q'$ in the current iteration 
	and add $Q'$ to $M$. It follows that $Q \subseteq M$ and that $Q$ is collapsed.

	Now consider the case $T= \emptyset$ (i.e., $\Diam{S} < 2 \gamma$).
	As $\ROut{S} \not= \emptyset$ we have that the set $X$ computed at Line~\ref{alg:symmec:attrB} is non-empty. 
	Note that $Q$ cannot contain a vertex in $X$ by Lemma~\ref{lem:attr_remove} 
	and the fact that $S$ is strongly connected by Lemma~\ref{lem:corr:0}. 
	Consequently, $Q$ is contained in one of the SCCs of $S \setminus X$. 
	Note that each such SCC has size strictly smaller than $S$ due to the fact that $X$ is non-empty.
	The claim then holds by the induction hypothesis.
	This completes the proof of Lemma~\ref{lem:allECInM}.
\end{proof}

By the above lemma, we have that the algorithm finds all ECs\@. 
We next show that all vertices added to $M$ are actually contained
in some EC\@.

\begin{lemma}\label{lem:MValid}
	The set $M$ is a subset of $M_S$, i.e., $M \subseteq \bigcup_{Q \in \EC{S}} Q$.
\end{lemma}
\begin{proof}
	We show the claim by induction over the size of $S$.
	For the base case, i.e., $|S| \leq 1$ the claim holds trivially because we return
	the empty set at Line~\ref{alg:symmec:ret0} and any set of size less than two only contains
	trivial ECs.
	For the inductive step, let $|S| > 1$. 
	For the return statement at Line~\ref{alg:symmec:ret1} we argue as follows:
	Because $S$ has no random vertices with edges out of $S$ 
	(considering $P'$, Line~\ref{alg:symmec:if1}), $S$ is a non-trivial EC\@. 
	Thus, in Line~\ref{alg:symmec:ret1} we correctly return $S$.

	For the case $T \not=\emptyset$ we argue as follows:
	Let $T$, $A$ be the separator and its attractor as computed in Line~\ref{alg:symmec:sep} and
	Line~\ref{alg:symmec:attrSep}. Note that for all SCCs $S_j$ in the graph $S \setminus A$
	we have $|S_j| < |S|$ because $|A| > 0$. Thus, by induction hypothesis, all vertices added 
	in the recursive call at Line~\ref{alg:symmec:rec1} are in $M_S$.
	We claim that for each iteration of the while loop in Line~\ref{alg:symmec:while2} 
	we add only non-trivial ECs to $M$:
	Let $v$ be the vertex we choose to remove from $T$ at Line~\ref{alg:symmec:pick}.
	Let $S_v$ be the SCC of $v$ computed at Line~\ref{alg:symmec:computescc2}.
	If $|S_v| = 1$ we continue to the next iteration without adding anything to $M$ and 
	the claim holds.
	Otherwise, there are two cases based on the 
	computation of the attractor $Z$ of random vertices 
	with edges out of $S_v$ at Line~\ref{alg:symmec:attrIncr}:
	If $Z$ contains $S_v$, we do not add vertices to $M$ and the claim holds.
	If $S_v \setminus Z \neq \emptyset$ we add the SCC $S'_v$ of $v$ in $S_v \setminus Z$ to $M$.
	It remains to show that the non-trivial SCC $S'_v$ as computed at Line~\ref{alg:symmec:snew} is a non-trivial EC\@:

	Recall that by Claim~\ref{claim:non-trivialEC} a maximal non-trivial EC within the SCC $S'$ at
	Line~\ref{alg:symmec:computescc2} has to contain $v$. When we remove the random attractor $Z$ of
	$\ROut{S'}$ then we either end up with the empty set or with a set $S' \gets S' \setminus Z$
	that has some non-trivial SCCs with no random outgoing edges. 
	Note that for $u \in S'_v \cap V_1$ we have $\Out{u} \cap S'_v \neq \emptyset$ (otherwise
	$u \in Z$) and for $u \in S'_v \cap V_R$ we must have $\Out{u} \subseteq S'_v$ (otherwise
	$u \in Z$). 
	These bottom SCCs are also ECs and, by
	the above, contain $v$. Thus there is a unique maximal bottom SCC and Line~\ref{alg:symmec:bottomSCC}
	returns that bottom SCC.
	
	Note that for all $v \in S'_v \cap V_1$ we have $\Out{v} \cap S'_v \neq \emptyset$, otherwise
	$v \in Z$. Also, for all $u \in S'_v \cap V_R$ we must have $\Out{u} \subseteq S'_v$, otherwise
	$u \in Z$. Thus, $|S'_v| > 1$ and all random edges again go to $S'_v$. 

    For the case $T =\emptyset$ note that each SCC found at Line~\ref{alg:symmec:while3} are of size strictly
	less than $|S|$ because $|X| \geq 1$. 
	Thus, by induction hypothesis, the vertices added at Line~\ref{alg:symmec:rec2}
	are in $M_S$.
\end{proof}

Lemma~\ref{lem:allECInM} and Lemma~\ref{lem:MValid} imply the correctness of 
$\Call{SymMec}{S,\gamma,P}$ and $\Call{SymbolicMec}{P',\gamma}$.

\begin{proposition}[Correctness]\label{prop:symmec:corr}
  Given an SCC $S$ of an MDP, $\Call{SymMec}{S,\gamma,P}$ returns the set $M_S=\bigcup_{Q \in \EC{S}} Q$, i.e., the set of vertices that are contained in a non-trivial MEC of $S$.
\end{proposition}
\begin{proposition}[Correctness]
 $\Call{SymbolicMec}{P',\gamma}$ returns the set of non-trivial MECs of an MDP $P'$.
\end{proposition}
\begin{proof}
	Due to Proposition~\ref{prop:symmec:corr} $M$ contains all vertices in nontrivial ECs.
	It remains to show that the nontrivial MECs are the SCCs of $M$. 
	By definition, each nontrivial EC is strongly connected. Towards a contradiction assume
	that a nontrivial MEC $Q$ is not an SCC by itself but part of some larger SCC $C$ of $M$. 
	But then, because $M$ contains only vertices in nontrivial ECs, we have that $C$ is the union of several ECs, i.e., 
	it is strongly connected and has no random outgoing edges and thus $C$ is an $EC$. This is in contradiction to $Q$ being a MEC\@.
	Thus, each nontrivial MEC is an SCC and because $M$ contains only vertices in nontrivial ECs the union of the MECs covers
	$M$. Thus there are no further SCCs.
\end{proof}

\subsection{Symbolic Operations Analysis}
We first bound the total number of symbolic operations for computing the separator $T$ at 
Line~\ref{alg:symmec:sep}, recursing upon the SCCs $S \setminus T$ at Line~\ref{alg:symmec:rec1} and 
adding the vertices in $T$ back to compute the rest of the ECs of $S$ at 
Lines~\ref{alg:symmec:while2}--\ref{alg:symmec:ret2} during all calls to $\Call{SymMec}{S,\gamma,P}$.
\begin{lemma}\label{lem:running_rec2}
	The total number of symbolic operations of Lines~\ref{alg:symmec:sep}--\ref{alg:symmec:ret2}
	in all calls to $\Call{SymMec}{S,\gamma,P}$ is in $O\left(\frac{n^2}{\lfloor \gamma /(2\log n)\rfloor}\right)$. 
\end{lemma}
\begin{proof}
	In Lemma~\ref{lem:runspacesep} we proved that computing the separator at 
	Line~\ref{alg:symmec:sep} takes $O(|S|)$ symbolic operations. 
	The same holds for computing the attractor at Line~\ref{alg:symmec:attrSep} due to 
	Lemma~\ref{lem:attrRunning} and computing the SCCs at Line~\ref{alg:symmec:while1}. 
	Each iteration of the while-loop at Line~\ref{alg:symmec:while2}
	takes $O(|S|)$ symbolic operations: Computing the SCC of $v$ twice can be done in $O(|S|)$ symbolic operations 
	due to Theorem~\ref{thm:SCCs}. Similarly, computing the random attractor at 
	Line~\ref{alg:symmec:attrIncr} is in $O(|S|)$ symbolic operations due to Lemma~\ref{lem:attrRunning}.
	The remaining lines can be done in a constant amount of symbolic operations.
	It remains to bound the symbolic operations of the while-loop at Line~\ref{alg:symmec:while1} where
	we call $\Call{SymMec}{S_j,\gamma,P}$ recursively at Line~\ref{alg:symmec:rec1} for each 
	SCC $S_j$ in $S \setminus A$. Note that the size of $T$ determines the number of iterations the while-loop 
	at Line~\ref{alg:symmec:while2} has, and how big the SCCs in $S \setminus A$ are. 
	We obtain that $|T|$ is of size at most
	$\frac{|S|}{\lfloor \gamma /(2 \log |S|)\rfloor}$ combining 
	Observation~\ref{obs:SepSize} and Lemma~\ref{lem:sepqual}.
	Due to the argument above the following equation bounds the 
	running time of Lines~\ref{alg:symmec:sep}--\ref{alg:symmec:ret2} for some constant $c$ which is
	greater than the number of constant symbolic operations in $\Call{SymMec}{C,\gamma,P}$ if $|S| \geq \gamma$.
	\begin{equation*}
		F(S)  \leq |T| \cdot |S| \cdot c + 
			\sum_{S_i \in \SCCs{S \setminus T}} F(S_i)
\end{equation*}
If $|S| < \gamma$ we only have the costs for computing the separator, i.e., $F(S) \leq c \cdot |S|$. 
 Next, we prove the bound in Claim~\ref{claim:fs}.
 	\begin{claim}\label{claim:fs}
 		$F(S) \in O\left(\frac{|S|^2}{\lfloor\gamma/(2 \log{|S|}) \rfloor} + |S|\right)$.
 	\end{claim}
 	\begin{proof}	
 		We prove the  inequality m by induction on the size of $S$.
 		That is we show $F(S) \leq \frac{|S|^2}{Z} c'$ where $Z = \lfloor\gamma/(2 \log |S|)\rfloor$ for some $c' > c$.
 		Obviously the inequality is true for $|S| < \gamma$, i.e., the base case is true.
 		For the inductive step, consider 
 		$F(S) \leq \frac{|S|^2}{Z}c'$ for $|S| \geq \gamma$.
  		Note that
 		\begin{align*}
			 \sum_{S_i \in \SCCs{S \setminus T}} F(S_i) &\leq \sum_{S_i \in \SCCs{S \setminus T}} c' \frac{|S_i|^2}{\lfloor\gamma/(2 \log |S_i|)\rfloor} + |S_i| \leq c'|S| + \sum_{S_i \in \SCCs{S \setminus T}} c' |S_i|^2/Z \\
			&\leq c'|S| + \sum_{S_i \in \SCCs{S \setminus T}} c' |S_i| \cdot (|S|-|T|Z) /Z \leq c'|S| + c'(|S|-|T|)(|S|-|T|Z) /Z
 		\end{align*}
 		The first inequality is due to the induction hypothesis and the third inequality is 
 		due to the fact that we have a $Z$-separator and Lemma~\ref{lem:sepqual}.
 		It remains to add $|S||T|c$:
 		\begin{align*}
 			F(S) &\leq c'|S| +  c'(|S|-|T|)(|S|-|T|Z) /Z + |S||T|c \leq c'|S| + \frac{c'}{Z} (|S|^2 -|S||T|-|T||S|Z+|T||S|) + |S||T|c\\
				 &= c'|S| + |S|^2/Zc' -|S||T|c' + |S||T|c \leq c'|S| + \frac{|S|^2}{Z}c'.
 		\end{align*}
 		The first inequality is due to the fact that $|T| \leq |S|/Z$ (Observation~\ref{obs:SepSize}).
 		This concludes our proof by induction of $F(S) \leq c' \cdot (\frac{|S|^2}{Z} + |S|)$.
 	\end{proof}
	Using this upper bound and $|S| \leq n$ we obtain a $O\left( n^2 / \lfloor\gamma/(2 \log{n}) \rfloor + n\right)$ bound which can be simplified to $O\left(n^2/\lfloor\gamma/(2 \log{n}) \rfloor\right)$ as $\gamma \leq n$.
\end{proof}

The second part of our analysis bounds the symbolic operations of the case when 
$\Diam{S} < 2\gamma$ at Lines~\ref{alg:symmec:attrB}--\ref{alg:symmec:rec2} and the work done from Line~\ref{alg:symmec:initM} to Line~\ref{alg:symmec:ret1}.

\begin{lemma}\label{lem:running_rec3}
	The total number of symbolic operations of Lines~\ref{alg:symmec:attrB}--\ref{alg:symmec:rec2}
	in all calls to $\Call{SymMec}{S,\gamma,P}$ is in $O(n \cdot \gamma + \frac{n^2}{\lfloor \gamma /(2 \log n)\rfloor})$. 
\end{lemma}
\begin{proof}
	If a vertex is in the set $X$ at Line~\ref{alg:symmec:attrB} it is not recursed upon or ever looked at again,
	thus we charge the symbolic operations of all attractor computations to the vertices in the attractor. 
	This adds up to a total of $O(n)$ symbolic operations by Lemma~\ref{lem:attrRunning}.
	Additionally, note that $|X|$ is non-empty, as otherwise $|S|$ is declared as EC in the 
	if-condition at Line~\ref{alg:symmec:if1}. Thus, Lines~\ref{alg:symmec:attrB}--\ref{alg:symmec:rec2}
	occur at most $n$ times.
	The number of symbolic operations for computing SCCs is in time $O(\sum_{C \in \SCCs{G}} (D_C +1))$ due to 
	Theorem~\ref{thm:SCCs}. We can distribute the costs to the SCCs according to the diameters of the SCCs.	
	We provide separate arguments for counting the costs for SCCs $S_j$ with  $\Diam{S_j} < 2 \gamma$ 
	and SCCs $S_j$ with  $\Diam{S_j} \geq 2 \gamma$.
    First, we consider the costs for SCCs $S_j$ with  $\Diam{S_j} < 2 \gamma$. 
    Computing the SCC $S_j$ costs $O(\gamma)$ operations and  
    the algorithm either terminates in the next step or at least one vertex is removed from the SCC via a separator or attractor.
    That is we have at most $n$ such SCCs computations and thus an overall cost of $O(n \gamma)$.
	Now we consider the costs for SCCs $S_j$ with  $\Diam{S_j} \geq 2 \gamma$. 
	Whenever computing such an SCC, we simply charge all its vertices for the costs of computing the SCC, 
	i.e., $O(1)$ for each vertex.
	For such an SCC the algorithm either terminates in the next step or a separator is computed and removed.
	We thus have that each vertex is charged again after at least $\lfloor \gamma /(2 \log |S|)\rfloor$
	many nodes are removed from its SCC\@. In total, each vertex is charged at most $\frac{|S|}{\lfloor \gamma /(2 \log |S|)\rfloor}$ many times.
	We get an $O(|S|\frac{|S|}{\lfloor \gamma /(2 \log |S|)\rfloor})$ upper bound for the SCCs with large diameter,
\end{proof}

Putting Lemma~\ref{lem:running_rec2} and Lemma~\ref{lem:running_rec3} together,
we obtain the $O(n \cdot \gamma + \frac{n^2}{\lfloor \gamma/(2\log n) \rfloor})$ bound for $\Call{SymMec}{S,\gamma,P}$ 
which  also applies to $\Call{SymbolicMEC}{P',\gamma}$ as the SCC-computations only require $O(n)$ operations.
\begin{proposition}\label{prop:symmec:time}\label{prop:metasymmec:time}
	$\Call{SymMec}{S,\!\gamma,\!P}$~and~$\Call{SymbolicMEC}{P',\!\gamma}$ both need $O(n \cdot \gamma + \frac{n^2}{\lfloor \gamma/(2\log n) \rfloor})$ symbolic operations.
\end{proposition}

\subsection{Symbolic Space} 
Symbolic space usage counted as the maximum number of sets (and not their size) at any point in time is a crucial metric and limiting factor of symbolic computation in practice~\cite{ClarkeGP99}.
In this section, we consider the symbolic space usage of $\Call{SymMec}{S,\gamma,P}$, highlight a key issue
and present a solution to the issue. 

\para{Key Issue.}
Even though $\Call{SymMec}{S,\gamma,P}$ beats the current best \emph{symbolic} Algorithm for computing
the MEC decomposition in the number of symbolic operations (current best: $O(n \sqrt{m})$, space: $O(\sqrt{n})$~\cite{CHLOT18}),
without further improvements, $\Call{SymMec}{S,\gamma,P}$ requires $O(n)$ symbolic space as we discuss in the following.
First, note that each call of $\Call{SymMec}{S,\gamma,P}$, when excluding the sets stored by recursive calls, only stores a constant number of sets and requires a logarithmic number of sets to execute $\SCCFind{}{\cdot}$.
That is, the recursion depth is the crucial factor here. 
As we show below, by the $\lfloor \gamma/(2\log n)\rfloor $-separator property, the recursion depth due to the case $T\not=\emptyset$ is  $O\left( n / \lfloor \gamma /(2\log n)\rfloor\right)$.
However, the case $T=\emptyset$ might lead to a recursion depth of $O(n)$ when 
in each iteration only a constant number of vertices is removed
and the diameter of the resulting SCC is still smaller than $2\gamma$.
As $\Call{SymMec}{S,\gamma,P}$ uses a constant amount of sets for each recursive call, 
it uses $\O(n)$ space in total.

\para{Reducing the symbolic space.}
We resolve the above space issue by modifying $\Call{SymMec}{S,\gamma,P}$ for the case $T=\emptyset$
(see Figure~\ref{alg:symmec_reduced_space}. ):
At the while loop at Line~\ref{alg:symmec:while3} we first consider the SCCs $S_j$ 
with less than $|S|/2$ vertices and recurse on them.
If there is an SCC $S'$ with more than $|S|/2$ vertices we process it at the end, i.e., we use one additional set to store that SCC until the SCC algorithm terminates.
As now all the computations of the current call to $\Call{SymMec}{S,\gamma,P}$ are done we can 
simply reuse the sets of the current calls to start the computation for $S'$. 
We do so by setting $S$ to $S'$ and continuing in the Line~1 of $\Call{SymMec}{S,\gamma,P}$.
Using that we only recurse on sets which are of size $\leq |S|/2$
and thus get a recursion depth of $O(\log n)$ for this case.
Moreover, the modified algorithm has the same computation operations as the original one and thus 
the bounds for the number of symbolic operations apply as well. 

\begin{figure}[t]
\begin{algorithmic}[1]
	\setalglineno{25}
	\EElse{} \Comment{$T = \emptyset$ and thus  $\Diam{S} < 2\gamma$}
	\State $X \gets \Attr{R}{P[S]}{\ROut{S}}$; $S_{tmp} = \emptyset$;\label{alg:symmec_reduced_space:attrB}
		\While{$S_j \gets \SCCFind{P}{S\setminus X, \emptyset}$}\label{alg:symmec_reduced_space:while3}
			\lIf{$S_j \geq |S|/2$ and $|S_j|>1$} $S_{tmp} \gets S_j$; \Continue;
			\State $M \gets M \cup \Call{SymMec}{S_j,M, \gamma,P}$\label{alg:symmec_reduced_space:rec2}
		\EndWhile	
		\lIf{$S_{tmp} \neq \emptyset$} $S\gets S_{tmp}$; goto Line~3;
		\State \Return{$M$}; \label{alg:symmec_reduced_space:ret3}
		\EndEElse{}
\end{algorithmic}
\caption{Reducing the symbolic space of the $\protect\Call{Symmec}{\cdot}$ procedure}\label{alg:symmec_reduced_space}
\end{figure}

\begin{lemma}\label{lem:recursiondepth}
	The maximum recursion depth of the modified algorithm is in  
	$O(\frac{n}{\lfloor \gamma /(2\log n)\rfloor} + \log n)$.
\end{lemma}
 \begin{proof}
     Consider the recursion occurring due to Line~\ref{alg:symmec:rec1}.
 	Because $T$ is a $\lfloor \gamma/(2\log n)\rfloor$-separator, an SCC in $S \setminus T$
 	contains at most $n-\lfloor \gamma/(2\log n)\rfloor \cdot |T| \leq n-\lfloor \gamma/(2\log n)\rfloor$
 	 vertices as $|T| \geq 1$. We determine how often we can remove
 	$\lfloor \gamma/(2\log n)\rfloor$ from $n$ until there are no vertices 
 	left which gives the recursion depth $k$.  It follows that $k=
 	\frac{n}{\lfloor \gamma /(2\log n)\rfloor}$.
 	Now consider the recursion occurring due to Line~\ref{alg:symmec:rec2}. In the modified version we have that $|S_j| < |S|/2$ 
 	and thus this kind of recursion is bounded by $O(\log n$).
 \end{proof}

$\Call{SymMec}{S,\gamma,P}$ has $O(n\gamma + n^2/\lfloor \gamma/(2\log n) \rfloor)$ many symbolic
operations due to Lemma~\ref{lem:running_rec2} and Lemma~\ref{lem:running_rec3} and the number of sets is
in $O(n/\lfloor \gamma /(2\log n)\rfloor + \log n)$ due to Lemma~\ref{lem:recursiondepth}.
In symbolic algorithms it is of particular interest to optimize symbolic space resources. Note that we obtain
a space-time trade-off when setting the parameter $\gamma$ such that $(2\sqrt{n}+ 2) \log n \leq
\gamma \leq n$.

For the symbolic space of $\Call{SymbolicMEC}{P',\gamma}$ notice that the SCC algorithms are in logarithmic symbolic space and the algorithm itself only needs to store the set $M$ and the current SCC\@. 
Thus, when we immediately output the computed MECs it only requires $O(\log n)$ additional space.

\begin{theorem}\label{thm:symbolicmecTST}
	The MEC decomposition of an MDP can be computed in $O\left(n^{2-\epsilon} \log n\right)$
	symbolic operations and with symbolic space $O\left(n^{\epsilon} \log n \right)$ for $0 < \epsilon \leq 0.5$.
\end{theorem}

By setting $\epsilon = 0.5$ we obtain 
that the MEC decomposition of an MDP can be computed in $\O(n\sqrt{n})$ symbolic operations and with symbolic space $\O(\sqrt{n})$.

\section{Symbolic Qualitative Analysis of Parity Objectives}
In this section we present symbolic algorithms for the qualitative analysis of parity objectives. 
To present the algorithms we first introduce further definitions and notation related to parity
objectives.

\paragraph{Plays and strategies}
An infinite path or a \emph{play} of an MDP $P$ is an infinite sequence $\omega = \langle v_0, v_1,
\dots \rangle$ of vertices such that $(v_k, v_{k+1}) \in E$ for all $k \in \mathbb{N}$.
We write $\Omega$ for the set of all plays.
A strategy for player~1 is a function $\sigma: V^* \cdot V_1 \rightarrow V$ that chooses the 
successor for all finite sequences $\vec{w} \in V^* \cdot V_1$ of vertices ending in a player-1
vertex (the sequence represents a prefix of a play). A strategy must respect the edge relation: for
all $\vec{w} \in V^*$ and $v \in V_1$ we have $(v, \sigma(\vec{w} \cdot v)) \in E$.
Player~1 follows the strategy $\sigma$ if, in each player-1 move, given that the current history of
the game is $\vec{w} \in V^* \cdot V_1$, she chooses the next vertex according to $\sigma(\vec{w})$.
We denote by $\Sigma$ the set of all strategies for player~1. A \emph{memoryless} player-1 strategy
does not depend on the history of the play but only on the current vertex: For all $\vec{w},
\vec{w'} \in V^*$ and for all $v\in V_1$ we have $\sigma(\vec{w} \cdot v) = \sigma(\vec{w'}\cdot
v)$. A memoryless strategy can be represented as a function $\sigma: V_1 \rightarrow V$. From now
on,
we consider only the class of memoryless strategies. Once a starting vertex $v \in V$ and a strategy
$\sigma \in \Sigma$ is fixed, the outcome of the MDP $P$ is a random walk $\omega_v^\sigma$ for which
the probabilities of events are uniquely defined. An \emph{event} $\A \subseteq \Omega$ is a
measurable set of plays. For a vertex $v \in V$ and an event $\A \subseteq \Omega$, we write
$\Pr_v^{\sigma}(\A)$ for the probability that a play belongs to $\A$ if the game starts from the
vertex $v$ and player~1 follows the strategy $\sigma$.

\paragraph{Objectives}
We define \emph{objectives} for player~1 as a set of plays $\Phi \subseteq
\Omega$. We say that a play $\omega$ satisfies the objective $\Phi$ if $\omega \in \Phi$. We
consider $\omega$-regular objectives~\cite{thomas97}, specified as parity conditions and reachability
objectives which are an important subset of $\omega$-regular objectives. 

\smallskip\noindent\emph{Reachability objectives.}
Given a set $T \subseteq V$ of ``target'' vertices, the reachability objective requires that
\emph{some} vertex of $T$ be visited. The set of winning plays is in $\Reach{T} = \{ \langle v_0,
	v_1, \dots \rangle \in \Omega \mid v_k \in T \text{ for some } k \geq 0 \}$.

\smallskip\noindent\emph{Parity objectives.}
A parity objective consists of a priority function which assigns an \emph{priority} to each vertex,
i.e, $p:V \mapsto \{0,1,\dots,2d\}$ where $d\in \mathbb{N}$. For a play $\omega = \langle v_0,v_1,
\dots \rangle \in \Omega$, we define $\Inf{\omega} = \{ v \in V \mid v_k = v \text { for infinitely
		many k }\}$ to be the set of vertices that occur infinitely often in $\sigma$. The
\emph{parity objective} is defined as the set of plays such that the minimum priority occurring
infinitely often is even, i.e., $\Parity{p} = \{ \omega \in \Omega \mid \min_{v \in
		\Inf{\omega}}{p(v)}
	\text { is even} \}$. 

\smallskip\noindent\emph{Qualitative analysis: almost-sure winning.}
Given a parity objective $\Parity{p}$, a strategy $\sigma \in \Sigma$ is \emph{almost-sure winning}
for player~1 from the vertex $v$ if $\Pr_v^\sigma(\Parity{p}) = 1$.
The \emph{almost-sure winning set} $\ASW{\Parity{p}}$ for player-1 is the set of vertices
from which player~1 has an almost-sure winning strategy. The qualitative analysis of MDPs
corresponds to the computation of the almost-sure winning set for $\Parity{p}$. It follows from the
results of~\cite{CY95,deAlfaroThesis} that for all MDPs and parity objectives, if there is an almost-sure winning
strategy, then there is a memoryless almost-sure winning strategy.

\begin{theorem}[\cite{CY95,deAlfaroThesis}]
	For all MDPs $P$, and all parity objectives $\Parity{p}$, there exists a memoryless strategy 
	$\sigma$ such that for all $v \in \ASW{\Parity{p}}$ we have $\Pr_v^\sigma(\Parity{p}) = 1$.
\end{theorem}

\subsection{Almost-sure Reachability.}
In this section we present a symbolic algorithm which computes reachability 
objectives $\Reach{T}$ in an MDP\@. 
The algorithm is a symbolic version of~\cite[Theorem 4.1]{CDHL16}.

\smallskip\noindent\emph{Graph Reachability.}
For a graph $G = (V,E)$ and a set of vertices $S$, the set 
$\GraphReach{G}{S}$ is the set of vertices $V$ that can reach a vertex of $S$ within $G$. We compute it 
by repeatedly calling $S \gets \Pre{E}{S}$ until we reach a fixed point.
In the worst case, we add one vertex in each such call and need $|\GraphReach{G}{S} \setminus S| +1
= O(n)$ 
many \Pre{E}{\cdot} operations to reach a fixed point.

\smallskip\noindent\emph{Algorithm Description.}
Given an MDP $P$ we compute the set $\ASW{\Reach{T}}$ as follows:
First, if player~1 can reach one vertex of a MEC he can reach all the vertices of a MEC and thus 
we can collapse each MEC $M$ to a vertex.
If $M$ contains a vertex of $T$, we set $T \gets T \cup M$, i.e. $T$ contains the vertex which
represents $M$.
$P' = (V',E',\langle V'_1,V'_2 \rangle, \delta')$ is the MDP where the MECs of $P$ are collapsed as described above.
Next, we compute the set of vertices $S$ which can reach $T$ in the graph induced by $P'$. 
A vertex in $V' \setminus S$ cannot reach $T$ almost-surely because there is no path to $T$.
Note that a play starting from a vertex in the random attractor $A$ of $V' \setminus S$ might also end up in $V' \setminus S$. We
thus remove $A$ from $V'$ to obtain the set $R$, where vertices can almost-surely
reach $T$ in $P'$. Finally, to transfer the result back to $P$ we include all MECs with a vertex in $R$. 

We implement $\Call{SymbolicMEC}{P',\gamma}$ to compute the MEC decomposition of $P$ but note that we could use any symbolic MEC algorithm. 
To minimize the extra space usage we also assume that the algorithm which computes the
MEC decomposition outputs one MECs after the
other instead of all MECs at once. 
Note that we can easily modify $\Call{SymbolicMEC}{P',\gamma}$ to output one MEC after the
other by iteratively returning each SCC found at Line~\ref{alg:metasymmec:while2}. 
Moreover, we only require logarithmic space to maintain the state of the SCC
algorithm~\cite{ChatterjeeDHL18}.

\begin{figure}
	\begin{algorithmic}[1]
		\Procedure{SymASReach}{$T, P$}
		\State $V' \gets V; E' \gets E; V'_1 \gets V_1; V'_R \gets V_R;  \delta' \gets \delta$;
		\State $P'=(V',E', \langle V'_1, V'_R \rangle, \delta')$;
		\For{$M \gets \Call{ComputeMECs}{P}$}
			\State$\Call{CollapseEC}{M,P'}$; 
			\lIf{$M \cap T \neq \emptyset$} $T \gets T \cup M$;
		\EndFor	
		\State $S \gets \GraphReach{P'}{T \cap V'}$; $A \gets \Attr{R}{P'}{V'\setminus S}$; $R \gets V' \setminus A$;
		\For{$M \gets \Call{ComputeMECs}{P}$}
			\lIf{$M \cap R \neq \emptyset$} $R \gets R \cup M$;
		\EndFor	
		\State \Return $R$;
		\EndProcedure
	\end{algorithmic}
	\caption{Computing $\ASW{\Reach{T}}$ symbolically}\label{alg:symasreach}
\end{figure}

We prove the following two propositions for $\Call{SymASReach}{T,P}$.
Let $P$ be an MDP, $T$ a set of vertices and $\mathbf{MEC}$ 
be the number of symbolic operations we need to compute the MEC decomposition.
Let $\aspace(\mathbf{MEC})$ denote the space of computing the MEC decomposition.

\begin{proposition}[Correctness~\cite{ChatterjeeDHL18}]\label{prop:asreachcorr}
	$\Call{SymASReach}{T,P}$ correctly computes the set $\ASW{\Reach{T}}$.
\end{proposition}

\begin{proposition}[Running time and Space]\label{prop:asreachrunning}
	The total number of symbolic operations of $\Call{SymASReach}{T,P}$ is in $O(\mathbf{MEC} +
	n)$. $\Call{SymASReach}{T,P}$ uses $\O(\aspace(\mathbf{MEC}))$ symbolic space.
\end{proposition}

Proposition~\ref{prop:asreachrunning} and Proposition~\ref{prop:asreachcorr} together with Theorem~\ref{thm:symbolicmecTST} yield the following theorem.

\begin{theorem}\label{thm:asreach}
	The set $\ASW{\Reach{T}}$ of an MDP can be computed with $\O(n^{2-\epsilon})$ 
	many symbolic operations and $\O(n^\epsilon)$ 
	symbolic space for $0 < \epsilon \leq 0.5$.
\end{theorem}

\subsection{Parity Objectives.}
In this section we consider the qualitative analysis of MDPs with parity objectives.
We present an algorithm for computing the winning region which is based on the algorithms we present
in the previous sections and the
algorithm presented in~\cite[Section 5]{CH11}.
The algorithm presented in~\cite[Section 5]{CH11} draws ideas from a hierarchical clustering technique~\cite{Tarjan82, KingKV01}.
Without loss of generality, we consider the parity objectives $\Parity{p}$ where 
$p: V \rightarrow \{0,1,\dots,2d\}$.
In the symbolic setting, instead of $p$, we get the sets $\P_{\geq i} = \{ v \in V \mid p(v) \geq i \}$ where $(1 \leq i \leq 2d)$ as part of the input.  
We abbreviate the family $\{ \P_{\geq i} \mid 1 \leq i \leq 2d\}$ as $(\P_{\geq k})_{1\leq k\leq 2d}$. 
Let $\P_{\leq m} = V \setminus \P_{\geq m+1}$ and  $\P_{m} = \P_{\geq m} \setminus \P_{\geq m+1}$. 
Given an MDP $P$, let $P_i$ denote the MDP obtained by removing $\Attr{R}{P}{\P_{\leq i-1}}$, the set
of vertices with priority less than $i$ and its random attractor.
A MEC $M$ is a \emph{winning MEC} in $P_i$ if there exists a vertex $u \in M$ such that $p(u) = i$ and $i$ is even, i.e., the smallest priority in the MEC is even. Let $\WE_i$ be the union of
vertices of winning maximal end-components in $P_i$, and let $\WE = \bigcup_{0\leq i \leq 2d} \WE_i$.
Lemma~\ref{lem:paritywinningmec} says that computing $\ASW{\Parity{p}}$ is equivalent to computing almost-sure
reachability of $\WE$. Intuitively, player 1 can infinitely often satisfy the parity condition after
reaching an end-component which satisfies the parity condition. 

\begin{lemma}[\cite{CH11}]\label{lem:paritywinningmec}
	Given an MDP $P$ we have $\ASW{\Parity{p}} = \ASW{\Reach{\WE}}$.
\end{lemma}

We describe in Figure~\ref{alg:winpec} how to compute $\WE$ symbolically.

\subsubsection{Algorithm Description.}
The algorithm uses a key idea which we describe first. 
Recall that $P_i$ denotes the MDP obtained by removing $\Attr{R}{P}{\P_{\leq i-1}}$, i.e., the set
of vertices with priority less than $i$ and its random attractor.

\smallskip\noindent\emph{Key Idea.}
If $u,v$ are in a MEC in $P_i$, then they are in the same MEC in $P_{i-1}$. 
The key idea implies that if a vertex is in a winning MEC of $P_i$, it is also in a winning MEC of
$P_{i-1}$.
Intuitively, this holds due to the following two facts: 
(1)~Because $P_{i-1}$ contains all edges and
vertices of $P_i$ the MECs $P_i$ are still strongly connected in $P_{i-1}$. 
(2)~Because $\Attr{R}{P}{\P_{\leq i-1}}$ makes sure that no MEC $M$ in $P_i$ 
has a random vertex with an edge leaving $P_i$ in $M$, the same is true for the set $M$ in $P_{i-1}$. 

We next present the recursive algorithm $\Call{WinPEC}{(\P_{\geq k})_{1\leq k \leq 2d},i,j,P}$ which,
for a MDP $P$, computes the set $\bigcup_{i \leq \ell \leq j} \WE_i$ of winning MECs for priorities between $i$ and $j$.
\begin{enumerate}
	\item \emph{Base Case:} If $j < i$, return $\emptyset$.
	\item Compute $m \gets \lceil (i+j)/2 \rceil$.
	\item Compute the MECs of $P_m = V \setminus \Attr{R}{P}{\P_{\leq m-1}}$ and for each MEC $M \in P_m$
		compute the minimal priority $min$ among all vertices in that MEC.
		\begin{itemize}
			\item If $min$ is even then add $M$ to the set $W$ of vertices in winning MECs.
			\item If $min$ is odd we recursively call $\Call{WinPEC}{(\P_{\geq k})_{1\leq k \leq
						2d},min+1,j,P^u}$ where $P^u$ is the sub-MDP containing only vertices and edges inside $M$.
				This call applies the key idea and refines the MECs of $P_m$ 
				and  computes the set $\bigcup_{min+1\leq \ell \leq j} \WE_\ell$.	   
		\end{itemize}
	\item Call $\Call{WinPEC}{(\P_{\geq k})_{1\leq k \leq 2d},i,m-1,P^\ell}$ where $P^\ell$ is the MDP where all MECs in $P_m$ are
		collapsed into a single vertex and thus only the edges outside the MECs of $P_m$ are considered.
		This call computes the set $\bigcup_{i\leq k\leq m-1} \WE_k$.
\end{enumerate}
The initial call is $\Call{WinPEC}{\P_{\geq	k})_{1\leq k \leq 2d},0,2d,P}$. 
Figure~\ref{alg:winpec} illustrates the formal version of the sketched algorithm.

\begin{figure}
	\begin{algorithmic}[1]
		\Procedure{WinPEC}{$(\P_{\geq k})_{1\leq k \leq 2d},i,j,P$}
			\State $W \gets \emptyset$;
			\lIf{$j<i$} \Return $W$;
			\State $m \gets \lceil (i+j)/2 \rceil$;
			\State $X_m \gets \Attr{R}{P}{\P_{\leq m -1}}$; \label{alg:winpec:attrxm}
			\State $Z_m \gets V \setminus X_m$; $E_m \gets E \cap (Z_m \times Z_m)$; 
			\State $P' \gets (Z_m,E_m,\langle V_1 \cap Z_m, V_R \cap Z_m\rangle, \delta')$;
			\For{$M \gets \Call{ComputeMECs}{P'}$}\label{alg:winpec:computemecs1}
				\State $min \gets minPriority(M)$;\label{alg:winpec:minpriority}
				\If{$min$ is even}
					\State$W \gets W \cup M$;
				\Else
					\State $V^u \gets M \setminus \Attr{R}{P}{\P_{min}}$;\label{alg:winpec:rmPmin}
					\State $P^u \gets (V^u,(V^u \times V^u) \cap E,\langle V_1 \cap V^u, V_R \cap
					V^u\rangle, \delta^u)$;
					\State $W \gets W \cup \Call{WinPEC}{(\P_{\geq k})_{1\leq k \leq	2d},min+1,j,P^u}$;\label{alg:winpec:rec1}
				\EndIf	
			\EndFor	
			\For{$M \gets \Call{ComputeMECs}{P'}$}\label{alg:winpec:computemecs2}
				\State $\Call{CollapseEC}{M,P};$\label{alg:winpec:collapse}	\Comment{MDP with MECs collapsed is $P^\ell$}
			\EndFor	
			\State $W \gets W \cup \Call{WinPEC}{(\P_{\geq k})_{1\leq k \leq 2d},i, m-1,P}$;\label{alg:winpec:rec2}
			\State \Return $W$;
		\EndProcedure
	\end{algorithmic}	
	\caption{Algorithm to compute $\WE$ of the MDP $P$ recursively}\label{alg:winpec}
\end{figure}

\subsubsection{Correctness and Number of Symbolic Steps.}
In this section we argue that $\Call{WinPEC}{\cdot}$ is correct and bound the number of symbolic
steps and the symbolic space usage. 
A key difference in the analysis of $\Call{WinPEC}{\cdot}$ and~\cite{CH11} 
is that we aim for a symbolic step bound that is independent from the number of edges in $P$
and, thus, we cannot use the argument from \cite{CH11} which charges the cost of each recursive call
to the edges of $P$.
The key argument in~\cite{CH11} is that the sets of edges in the different branches of the recursions do not overlap.
For vertices it is not that simple, as we do not entirely remove vertices that appear in a MEC but merge the MEC and represent it by a single vertex.
That is, a vertex can appear in both $P^\ell$ and in $P^u$ corresponding to the MEC. 
In order to accomplish our symbolic step bound we adjusted the algorithm.
At Line~\ref{alg:winpec:rmPmin} we \emph{always} remove the minimum priority vertices
instead of removing the vertices with priority $m$ to ensure that we remove at least one vertex.
Intuitively, by always removing at least one vertex from a MEC we ensure that the total number of vertices processed
at each recursion level does not grow.
Note that these changes of the algorithm do not affect the  correctness argument of~\cite{CH11} as we always compute the same sets $\WE_m$ 
but avoid calls to $\Call{WinPEC}{\cdot}$ with no progress on some MECs.

\begin{proposition}[Correctness]\label{prop:parity:corr}
	$\Call{WinPEC}{\cdot}$ returns the set of winning end-components $\WE$.
\end{proposition}
\begin{proof}
	The correctness of the algorithm is by induction on $j-i$
	for the induction hypothesis 
	$\bigcup_{i \leq \ell \leq j} \WE_\ell \subseteq \Call{WinPEC}{(\P_{\geq k})_{1\leq k \leq
			2d},i,j,P} \subseteq \bigcup_{1 \leq \ell \leq 2d} \WE_\ell$. 
	
	First consider the induction base cases:
	If $j>i$, the algorithm correctly returns the empty set.
	Next, consider the induction step. 
	Assume that the results hold for $j-i \leq k$, and we consider $j-i = k+1$. 
	If $m$ is even, then
	\begin{equation*}
		\bigcup_{i \leq k \leq j} \WE_k = \WE_{m} \cup \bigcup_{i \leq k \leq m-1} \WE_k \cup
		\bigcup_{m+1 \leq k \leq j} \WE_k,
	\end{equation*}
	otherwise ($m$ is odd), then
	\begin{equation*}
		\bigcup_{i \leq k \leq j} \WE_k = \bigcup_{i \leq k \leq m-1} \WE_k \cup \bigcup_{m+1 \leq k \leq j} \WE_k,
	\end{equation*}	
	Consider an arbitrary winning MEC $M_k$ in $P_k$, i.e., the lowest even priority is $k$. 
	We consider the following cases.
	\begin{enumerate}
		\item For all $k \geq m$ we have that $M_k$ is contained in a MEC $M_m$ of $P_m$. 
			Additionally, no random vertex in $M_m$ can have an edge leaving $M_m$ and thus no
			random vertex in $M_k$ can have a random edge leaving $M_m$. 
			Moreover, for the minimum priority $min$ of $M_m$ we have $k \geq min \geq m$.
			If $min$ is even then $M_m$ is itself winning and thus $M_k \subseteq
			\Call{WinPEC}{(\P_{\geq k})_{1\leq k \leq 2d},i,j,P}$ and $M_m \subseteq \bigcup_{1 \leq \ell
				\leq 2d} \WE_\ell$ (note that it might be that $min > j$).
			
			If $min$ is odd 
			we have that $M_k$ is a winning MEC of $P$ iff 
			it is a winning MEC of $P^u$ and thus by the induction
			hypothesis $M_k \subseteq \Call{WinPEC}{(\P_{\geq k})_{1\leq k \leq
					2d},min+1,j,P^u}$.
			It follows that also $M_k \subseteq \Call{WinPEC}{(\P_{\geq k})_{1\leq k \leq 2d},i,j,P}$.
		\item For $k < m$ consider a MEC $M_k$ in $P_k$. If $M_k$ contains a vertex $v$ that
			belongs to a MEC $M_m$ of $P_m$, then $M_m \subset M_k$ (i.e., all vertices of the
			MEC in $P_m$ of $v$ also belong to $M_k$ and $M_k$ has at least one additional vertex
			with priority $<m$).
			We thus have that for $k \leq m$  the winning MECs $M_k$ in $P_k$ are in one-to-one correspondence
			with the winning MECs $M'_k$ of the modified MDP where all MECs of $P_m$ are collapsed.
			From the induction hypothesis 
			it follows that $\bigcup_{i\leq k \leq m-1} \WE_k = \Call{WinPEC}{(\P_{\geq k})_{1\leq k
					\leq 2d},i, m-1,P}$
	\end{enumerate}
	
	Hence we have that $\bigcup_{i \leq k \leq j} \WE_k  \subseteq \Call{WinPEC}{(\P_{\geq
			k})_{1\leq k \leq 2d},i,j,P} \subseteq  
	\bigcup_{1 \leq \ell \leq 2d} \WE_\ell$. The statements follows from setting $i=0$ and $j=2d$.
\end{proof}

\begin{proposition}[Symbolic Steps]\label{prop:parity:symbsteps}
	The total number of symbolic operations for $\Call{WinPEC}{\cdot}$ is $O(\mathbf{MEC} \cdot \log d)$ for
	$0 < \epsilon \leq 0.5$.
\end{proposition}

\begin{proof}
	Given an MDP $P$ with $n$ vertices and $d$ priorities, let us denote by
	$T(n,x)$ the number of symbolic steps of $\Call{WinPEC}{\cdot}$ at recursion depth $x$ 
	and with $T_M(n)$ the number of symbolic steps incurred by the symbolic MEC Algorithm. 
	As shown in~\cite{CH11}, note that the recursion depth of $\Call{WinPEC}{\cdot}$ 
	is in $O(\log d)$ because we recursively consider either $(min + 1,j)$ or $(i, m-1)$ where $min
	\geq m = \lceil(i+j/2)\rceil$ until $j>i$, where $j=2d$ initially.
	First, we argue that there exists $c > 0$ such that
	$
	T(n,x) \leq c \cdot T_M(n) + (\sum_{i=1,\dots,t} T\left(n_i, x-1\right)) \\ +T(n- (\sum_{i=1,\dots,t}
	n_i) + t , x-1)  \text{ if } x > 1$ and	$T(n,0) \leq c. 
	$
	The attractors computed at Line~\ref{alg:winpec:attrxm} and Line~\ref{alg:winpec:rmPmin} 
	can be done in $O(n)$ symbolic steps as the set of vertices in the attractors are all disjunct.
	Clearly, this is cheaper than computing the MEC decomposition. 
	To extract the minimum priority of a set of nodes $X
	\subseteq V$ we apply a binary search procedure which takes $O(\log d)$ symbolic steps at Line~\ref{alg:winpec:minpriority}.
	Note that when $\log d > n$ we cannot charge the cost to computing the MEC decomposition.
	Thus, we argue in Claim~\ref{claim:minpriority} that the total number of symbolic steps for Line~\ref{alg:winpec:minpriority} in $\Call{WinPEC}{\cdot}$ is less than $O(n \log d)$.
	The rest of the symbolic steps in $\Call{WinPEC}{\cdot}$, (except the recursive calls and
	computing the MEC decomposition) in $\Call{WinPEC}{\cdot}$ can be done in a constant amount of symbolic steps.
	Note that when $x = 0$, i.e., in the case $j < i$, we only need a constant amount of symbolic steps.	
	Let $t$ be the number of MECS in $P'$. 
	When $x > 0$, consider the following argumentation for the
	number of symbolic steps of the recursive calls:
	\begin{itemize}
		\item $\Call{WinPEC}{(\P_{\geq k})_{1\leq k \leq 2d},min,j,P^u} $:
			We perform the recursive call for each MEC $M_i \in P'$ $(1 \leq i \leq t)$ 
			where the vertex with minimum priority is odd.
			The total cost incurred by all such recursive calls is $\sum_{i=1,\dots,t} T(n_i,
			x-1)$ where $n_i \leq |M_i|-1$ because we always remove the
			vertices with minimum priority at Line~\ref{alg:winpec:rmPmin}.
		\item $\Call{WinPEC}{(\P_{\geq k})_{1\leq k \leq 2d},i, m-1,P^\ell}$:
			$P^\ell$ consists of the vertices representing the collapsed MECs, the vertices not
			in $P'$ and the vertices which are not in a MEC of $P'$. 
			The number of vertices in $P^\ell$ is thus $n_\ell = n - \sum_{i=1,\dots,t} |M_i| + t$
			and we obtain $T(n_\ell,x-1)$. 
	\end{itemize}
	Note that $\sum_{i=1, \dots, t} n_i + n_\ell 
	\leq n$.
	We choose $c$ such that $c T_M(n)$ is greater than the number of symbolic steps for computing
	the MECs twice and the rest of the work in the current iteration of $\Call{WinPEC}{\cdot}$.
	It is straightforward to show that $T(n,d) = O(T_M(n) \log d)$.
	
	The following claim shows that the total number of symbolic steps incurred by
	Line~\ref{alg:winpec:minpriority} for all calls to $\Call{WinPEC}{\cdot}$ is only $O(n \log d)$.
	\begin{claim}\label{claim:minpriority}
		The total amount of symbolic steps used by Line~\ref{alg:winpec:minpriority} is in
		$O(n \log d)$.
	\end{claim}
	\begin{proof}
		To obtain the set of vertices with minimum priority from a set of vertices $X \subseteq V$ the function
		$\minPriority{X}$ performs a binary search using the sets $(\P_{\geq k})_{1\leq k \leq 2d}$.
		This can be done in $O(\log d)$ many symbolic steps.
		To prove that the number of symbolic steps used by Line~\ref{alg:winpec:minpriority} in
		total is in
		$O(n \log d)$ note that each time the function is performed we either:
		(i) Remove all
		vertices in $M$, and we never perform the function on the vertices in $M$ again. We charge the cost to an arbitrary vertex in $M$.
		(ii) Remove at least one vertex at Line~\ref{alg:winpec:rmPmin} and we never perform
				the function on a MEC containing this vertex again. We charge the cost to this vertex.
		As there are only $n$ vertices we obtain that the total amount of symbolic steps used by
		Line~\ref{alg:winpec:minpriority} is in $O(n \log d)$.
	\end{proof}
	The symb. bound of $\Call{WinPEC}{\cdot}$ follows by Claim~\ref{claim:minpriority}.
\end{proof}
	
	\begin{proposition}\label{prop:parity:symbspace}
		$\Call{WinPEC}{\cdot}$ uses $O(\aspace(\mathbf{MEC}) + \log n \log d)$ space.
	\end{proposition}
	\begin{proof}
		Let $P$ be an MDP with $n$ vertices and $p$ a parity objective with $d$ priorities.
		We denote with $\aspace(\mathbf{MEC})$ the symbolic space used by the algorithm that computes the MEC decomposition.
		Observe that all computation steps in $\Call{WinPEC}{\cdot}$ need constant space except
		for the recursions and computing the MEC decomposition. 
		Both at Line~\ref{alg:winpec:computemecs1} and Line~\ref{alg:winpec:computemecs2}
		we first compute the MEC decomposition and then, to minimize extra space, 
		we output one MEC after the other by returning each SCC found given the set of vertices in nontrivial MECs.
		Note that we only require logarithmic space to maintain the state of the SCC algorithm~\cite{ChatterjeeDHL18}.  
		As argued in~\cite{CH11} the recursion depth of $\Call{WinPEC}{\cdot}$ is $O(\log d)$.
		Thus, we need $O(\log n \log d)$ space for maintaining the state of the SCC algorithm 	
		at Line~\ref{alg:winpec:attrxm} until we reach a leaf of the recursion tree. 
		At each	recursive call, we need additive $O(\aspace(\mathbf{MEC}))$ space to compute the MEC
		decomposition of $P$. The claimed space bound follows.
	\end{proof}

	Given an MDP, we first compute the set $\WE$ with $\Call{WinPEC}{\cdot}$ which is correct
	due to Proposition~\ref{prop:parity:corr}.
	We instantiate $\mathbf{MEC}$ and $\aspace(\mathbf{MEC})$ in Proposition~\ref{prop:parity:symbsteps} and
	Proposition~\ref{prop:parity:symbspace} respectively with Theorem~\ref{thm:symbolicmecTST} and thus need 
	$O(n^{2-\epsilon} \log n \log d)$ symbolic steps and $O(n^\epsilon \log n + \log n \log d)$
	(where $0 < \epsilon \leq 0.5$)	symbolic space for computing $\WE$.
	Then, we compute almost-sure reachability of $\WE$ with Theorem~\ref{thm:asreach}. 
	Finally, using Lemma~\ref{lem:paritywinningmec} we obtain the following theorem.
	\begin{theorem}
		The set $\ASW{\Parity{p}}$ of an MDP $P$ can be computed with $\O(n^{2-\epsilon})$ 
		many symbolic operations and $\O(n^{\epsilon})$ symbolic space for $0 < \epsilon \leq 0.5$.
	\end{theorem}

\section{Conclusion}
We present a faster symbolic algorithm for the MEC decomposition and we 
improve the fastest symbolic algorithm for verifying MDPs with $\omega$-regular properties. There are
several interesting directions for future work. On the practical side, implementations and
experiments with case studies is an interesting direction.
On the theoretical side, improving upon the $\O(n^{1.5})$ bound for MECs is an interesting open
question which would also, using our work, improve the presented algorithm for verifying
$\omega$-regular properties of MDPs.

\section*{Acknowledgements}
The authors are grateful to the anonymous referees for their valuable comments.
A.\ S.\ is fully supported by the Vienna Science and Technology Fund (WWTF) through project ICT15--003. K.\ C.\ is supported by the 
Austrian Science Fund (FWF) NFN Grant No S11407-N23 (RiSE/SHiNE) and by the ERC CoG 863818 (ForM-SMArt). 
For M.\ H.\ the research leading to these results has received funding from the European Research Council under the European Union’s Seventh 
Framework Programme (FP/2007--2013) / ERC Grant Agreement no. 340506.
\bibliographystyle{abbrv}
\clearpage
\bibliography{symbolicMEC}
\end{document}